\documentclass[12pt,letterpaper]{article}

\newif\ifprint

\usepackage{booktabs}

\usepackage{appendix}

\usepackage{latexsym}
\usepackage{amssymb,amsmath, bm}
\usepackage{mathtools}
\usepackage{centernot}
\usepackage{graphicx}
\usepackage{marvosym}
\usepackage{multirow}
\usepackage{setspace}
\usepackage{caption}
\usepackage{subcaption}
\usepackage{scalerel}
\usepackage[bottom]{footmisc}

\usepackage{tikz}
\usepackage{inputenc}
\usetikzlibrary{shapes,arrows,trees,fit,positioning,shapes.misc}

\usepackage{natbib}
\usepackage{color}
\usepackage[bookmarksopen=true, bookmarksnumbered=true,
pdfstartview=FitH, breaklinks=true, urlbordercolor={0 1 0}, citebordercolor={0 0 1}]{hyperref}
\usepackage{cleveref}

\usepackage[short,nocomma]{optidef}

\usepackage{dcolumn}
\newcolumntype{.}{D{.}{.}{-1}}
\newcolumntype{d}[1]{D{.}{.}{#1}}

\newcolumntype{C}[1]{>{\centering\arraybackslash}p{#1}}
\newcolumntype{L}[1]{>{\raggedright\arraybackslash}p{#1}}
\newcolumntype{R}[1]{>{\raggedleft\arraybackslash}p{#1}}

\usepackage{theorem}
\theoremstyle{plain}
\theoremheaderfont{\scshape}
\newtheorem{theorem}{Theorem}

\newtheorem{assumption}{Assumption}
\newtheorem{lemma}{Lemma}

\newcommand{\qed}{\hfill \ensuremath{\Box}}

\DeclarePairedDelimiterX{\ip}[2]{\langle}{\rangle}{#1,\, #2}
\DeclarePairedDelimiter\abs{\lvert}{\rvert}
\DeclarePairedDelimiter\norm{\lVert}{\rVert}

\newenvironment{proof}{\vspace{1ex}\noindent{\bf Proof}\hspace{0.5em}}
{\hfill\qed\vspace{1ex}}

\makeatletter
\let\oldabs\abs
\def\abs{\@ifstar{\oldabs}{\oldabs*}}
\let\oldnorm\norm
\def\norm{\@ifstar{\oldnorm}{\oldnorm*}}
\makeatother

\allowdisplaybreaks
\usepackage{rotating}

\usepackage{arydshln}

\usepackage[compact]{titlesec}

\usepackage{rotating}
\usepackage{pdflscape}

\begin{document}
	
	\newcommand\ud{\mathrm{d}}
	\newcommand\dist{\buildrel\rm d\over\sim}
	\newcommand\ind{\stackrel{\rm indep.}{\sim}}
	\newcommand\iid{\stackrel{\rm i.i.d.}{\sim}}
	\newcommand\logit{{\rm logit}}
	\renewcommand\r{\right}
	\renewcommand\l{\left}
	\newcommand\R{\mathbb{R}}
	\newcommand\cY{\mathcal{Y}}
	\newcommand\cZ{\mathcal{Z}}
	\newcommand\E{\mathbb{E}}
	\newcommand\V{\mathbb{V}}
	\newcommand\cA{\mathcal{A}}
	\newcommand\cC{\mathcal{C}}
	\newcommand\cD{\mathcal{D}}
	\newcommand\cE{\mathcal{E}}
	\newcommand\cM{\mathcal{M}}
	\newcommand\cN{\mathcal{N}}
	\newcommand\cT{\mathcal{T}}
	\newcommand\cTN{\mathcal{TN}}
	\newcommand{\Ip}{\mathcal{I}_+}
	\newcommand{\In}{\mathcal{I}_-}
	\newcommand\cX{\mathcal{X}}
	\newcommand\bA{\boldsymbol{A}}
	\newcommand\bC{\boldsymbol{C}}
	\newcommand\bI{\boldsymbol{I}}
	\newcommand\bK{\boldsymbol{K}}
	\newcommand\bP{\boldsymbol{P}}
	\newcommand\bQ{\boldsymbol{Q}}
	\newcommand\bU{\boldsymbol{U}}
	\newcommand\bD{\boldsymbol{D}}
	\newcommand\bS{\boldsymbol{S}}
	\newcommand\bX{\boldsymbol{X}}
	\newcommand\bW{\boldsymbol{W}}
	\newcommand\bM{\boldsymbol{M}}
	\newcommand\bZ{\boldsymbol{Z}}
	\newcommand\bY{\boldsymbol{Y}}
	\newcommand\balpha{\boldsymbol{\alpha}}
	\newcommand\bbeta{\boldsymbol{\beta}}
	\newcommand\bgamma{\boldsymbol{\gamma}}
	\newcommand\bdelta{\boldsymbol{\delta}}
	\newcommand\beps{\boldsymbol{\epsilon}}
	\newcommand\btheta{\boldsymbol{\theta}}
	\newcommand\bmu{\boldsymbol{\mu}}
	\newcommand\bxi{\boldsymbol{\xi}}
	\newcommand\bs{\boldsymbol{s}}
	\newcommand\bt{\boldsymbol{t}}
	\newcommand\bv{\boldsymbol{v}}
	\newcommand\bw{\boldsymbol{w}}
	\newcommand\bx{\boldsymbol{x}}
	\newcommand\by{\boldsymbol{y}}
	\newcommand\bone{\mathbf{1}}
	\newcommand\bzero{\mathbf{0}}
	\newcommand{\argmax}{\operatornamewithlimits{argmax}}
	\newcommand\ci{\perp\!\!\!\perp}
	
	\newcommand\spacingset[1]{\renewcommand{\baselinestretch}%
		{#1}\small\normalsize}
	
	\spacingset{1.1}
	
	
	
	{\title{Estimating Average Treatment Effects with Support
            Vector Machines\thanks{We thank Brian Lee for making his
              simulation code available and Francesca Dominici, Chad
              Hazlett, Gary King, Jose Zubizarreta and seminar
              participants at Institute for Quantitative Social
              Science, Harvard University for helpful discussions.  An
              anonymous reviewer of the Alexander and Diviya Magaro
              Peer Pre-Review Program for insightful comments.  Imai
              thanks the Sloan Foundation (\# 2020--13946) for
              financial support.}}
		
		\author{Alexander Tarr\thanks{PhD. Candidate, Department of Electrical Engineering, Princeton University, Princeton, NJ 08544. Email:\href{mailto:atarr@princeton.edu}{atarr@princeton.edu}} \and Kosuke Imai\thanks{Professor, Department of Government and
				Department of Statistics, Harvard University, Cambridge, MA
				02138. Phone: 617--384--6778, Email:
				\href{mailto:Imai@Harvard.Edu}{Imai@Harvard.Edu}, URL:
				\href{https://imai.fas.harvard.edu}{https://imai.fas.harvard.edu}}}
		
		\date{\today}
		
		\maketitle
		\thispagestyle{empty}
	}
	
	\begin{abstract}
		Support vector machine (SVM) is one of the most popular
		classification algorithms in the machine learning literature.  We
		demonstrate that SVM can be used to balance covariates and estimate
		average causal effects under the unconfoundedness assumption.
		Specifically, we adapt the SVM classifier as a kernel-based
		weighting procedure that minimizes the maximum mean discrepancy
		between the treatment and control groups while simultaneously
		maximizing effective sample size.  We also show that SVM is a
		continuous relaxation of the quadratic integer program for computing
		the largest balanced subset, establishing its direct relation to the
		cardinality matching method.  Another important feature of SVM is
		that the regularization parameter controls the trade-off between
		covariate balance and effective sample size.  As a result, the
		existing SVM path algorithm can be used to compute the
		balance-sample size frontier.  We characterize the bias of causal
		effect estimation arising from this trade-off, connecting the
		proposed SVM procedure to the existing kernel balancing methods.
		Finally, we conduct simulation and empirical studies to evaluate the
		performance of the proposed methodology and find that SVM is
		competitive with the state-of-the-art covariate balancing methods.
	\end{abstract}
	

	\clearpage
	\spacingset{1.83}
	\pagenumbering{arabic} 
	
\section{Introduction}

Estimating causal effects in an observational study is complicated by
the lack of randomization in treatment assignment, which may lead to
confounding bias. The standard approach is to weight observations such
that the empirical distribution of observed covariates is similar
between the treatment and control groups \citep[see
e.g.,][]{lunc:davi:04,rubi:06,ho:imai:king:stua:07,stua:10}.
Researchers then estimate the causal effects using the weighted sample
while assuming the absence of unobserved confounders.  Recently, a
large number of weighting methods have been proposed to directly
optimize covariate balance for causal effect estimation
\citep[e.g.,][]{hainmueller2012entropy,imai:ratk:14,zubizarreta2015stable,chan:yam:zhan:16,athe:imbe:wage:18,li:morg:zasl:18,
  wong2018kernel,zhao:19, hazlett2018kernel,
  kallus2020generalized,ning:peng:imai:20,tan:20}.

This paper provides a new insight to this fast growing literature on
covariate balancing by demonstrating that support vector machine
(SVM), which is one of the most popular classification algorithms in
machine learning \citep{cortes1995support,scholkopf2002learning}, can
be used to balance covariates and estimate the average treatment
effect under the standard unconfoundedness assumption.  We adapt the
dual coefficients for the soft-margin SVM classifier as kernel
balancing weights that minimizes the maximum mean discrepancy
\citep{gretton2007kernel} between the treatment and control groups
while simultaneously maximizing effective sample size.  The resulting
weights are bounded, leading to stable causal effect estimation.  As
SVM has been extensively studied and widely used, we can exploit its
well-known theoretical properties and highly optimized
implementations.

All matching and weighting methods face the same trade-off between
effective sample size and covariate balance with a better balance
typically leading to a smaller sample size.  We show that SVM directly
addresses this fundamental trade-off.  Specifically, the dual
optimization problem for SVM computes a set of balancing weights as
the dual coefficients while yielding the support vectors that comprise
a largest balanced subset.  In addition, the regularization parameter
of SVM controls the trade-off between sample size and covariate
balance.  We use existing path algorithms
\citep{hastie2004entire,sentelle2016simple} to efficiently
characterize the balance-sample size frontier \citep{king2017balance}.
We analyze how this trade-off between sample size and covariate
balance affect causal effect estimation.

We are not the first to realize the connection between SVM and
covariate balancing.  \cite{ratkovic2014balancing} points out that the
hinge-loss function of the primal optimization problem for SVM has a
first-order condition, which leads to balanced covariate {\it sums}
amongst the support vectors.  Instead, we show that the dual form of
the SVM optimization problem leads to the covariate {\it mean}
balance.  In addition, \cite{ghosh2018relaxed} notes the relationship
between the SVM margin and the region of covariate overlap. The author
argues that the support vectors correspond to observations lying in
the intersection of the convex hulls for the treated and control
samples \citep{king2006dangers}. In contrast, we show that the SVM
dual coefficients can be used to obtain weights for causal effect
estimation.  Furthermore, neither of these previous works studies the
relationship between the regularization parameter of SVM and the
trade-off between covariate balance and effective sample size.

The proposed methodology is also related to several covariate
balancing methods.  First, we establish that the SVM dual optimization
problem can be seen as a continuous relaxation of the quadratic
integer program for computing the largest balanced subset, which is
related to cardinality matching \citep{zubizarreta2014matching}.
Second, weighting with SVM can be interpreted as a kernel-based
covariate balancing method.  Several researchers have recently
developed weighting methods to balance functions in a reproducing
kernel Hilbert space (RKHS)
\citep{wong2018kernel,hazlett2018kernel,kallus2020generalized}.  SVM
shares the advantage of these methods that it can balance a general
class of functions and easily accommodate non-linearity and
non-additivity in the conditional expectation functions for the
outcomes.  In particular, we show that SVM fits into the kernel
optimal matching framework \citep{kallus2020generalized}. Unlike these
covariate balancing methods, however, we can exploit the existing path
algorithms for SVM to compute the set of solutions over the entire
regularization path with comparable complexity to computing a single
solution.  Finally, we also show that a variant of soft-margin SVM is
related to stable balancing weights \citep{zubizarreta2015stable}.

The rest of the paper is structured as follows.  In
Section~\ref{sec:methods}, we present our methodological results.  In
Section~\ref{sec:sims}, we conduct simulation studies to compare the
performance of SVM with that of the related covariate balancing
methods.  Lastly, in Section~\ref{sec:rhc}, we apply SVM to the data
from the right heart catheterization study
\citep{connors1996effectiveness}.


\section{Methodology}
\label{sec:methods}

In this section, we establish several properties of SVM as a covariate
balancing method.  First, the SVM dual can be viewed as a regularized
optimization problem that minimizes the maximum mean discrepancy
(MMD).  Second, we show that SVM is a relaxation of the largest
balanced subset selection problem.  Third, the regularization path
algorithm for SVM can be used to obtain the balance-sample size
frontier. Lastly, we discuss how to use SVM for causal effect
estimation and compare SVM to existing kernel balancing methods.

\subsection{Setup and Assumptions}

Suppose that we observe a simple random sample of $N$ units from a
super-population of interest, $\mathcal{P}$.  Denote the observed data
by $\mathcal{D} = \{\bX_i, Y_i, T_i\}_{i=1}^N$ where $\bX_i \in \cX$
represents a $p$-dimensional vector of covariates, $Y_i$ is the
outcome variable, and $T_i$ is a binary treatment assignment variable
that is equal to 1 if unit $i$ is treated and 0 otherwise.  We define
the index sets for the treatment and control groups as
$\mathcal{T} = \{i : T_i = 1\}$ and $\mathcal{C} = \{i : T_i = 0\}$
with the group sizes equal to $n_T = \abs{\mathcal{T}}$ and
$n_C = \abs{\mathcal{C}}$, respectively.  Finally, we define the
observed outcome as $Y_i = T_iY_i(1) + (1-T_i)Y_i(0)$ where $Y_i(1)$
and $Y_i(0)$ are the potential outcomes under treatment and control
conditions, respectively.

This notation implies the Stable Unit Treatment Value Assumption
(SUTVA) --- no interference between units and the same version of the
treatment \citep{rubi:90}.  Furthermore, we maintain the following
standard identification assumptions throughout this paper.
\begin{assumption}[Unconfoundedness]
	\label{assn:uncon} \spacingset{1} The potential outcomes
	$\{Y_i(1), Y_i(0)\}$ are indep\-endent of the treatment assignments
	$T_i$ conditional on the covariates $\bX_i$.  That is, for all $\bx
	\in \cX$, we have
	\begin{equation*}
		\{Y_i(1), Y_i(0)\} \ \perp\!\!\!\perp \ T_i \mid \bX_i  = \bx.
	\end{equation*}
\end{assumption}

\begin{assumption}[Overlap]
	\label{assn:over} \spacingset{1} For all $\bx \in \mathcal{X}$, the
	propensity score $e(\bx)=\Pr(T_i = 1 \mid \bX_i = \bx)$ is bounded
	away from 0 and 1, i.e.,
	\begin{equation*}
		0 < e(\bx) < 1.
	\end{equation*} 
\end{assumption}

Following the convention of classification methods, it is convenient
to define the following transformed treatment variable, which is equal
to either $-1$ or $1$,
\begin{equation}
	W_i \ = \ 2 T_i - 1 \ \in \ \{-1, 1\}.
\end{equation}
In addition, for $t = 0,1$, we define the conditional
expectation functions, disturbances, and conditional variance functions
as,
\begin{equation*}
	\E(Y_i(t) \mid \bX_i) = f_t(\bX_i), \quad \epsilon_i(t) = Y_i(t) - f_t(\bX_i), \quad \sigma_{t}^2(\bX_i) = \V(Y_i(t) \mid \bX_i).
\end{equation*}
Thus, we have $\E(\epsilon_i(t) \mid \bX_i) = 0$ for $t=0,1$. Lastly,
let $\mathcal{H}_K$ denote a reproducing kernel Hilbert space (RKHS)
with norm $\norm{\cdot}_{\mathcal{H}_K}$ and kernel
$K(\bX_i,\bX_j) = \ip*{\phi(\bX_i)}{\phi(\bX_j)}_{\mathcal{H}_K}$,
where $\phi: \R^p \mapsto \mathcal{H}_K$ is a feature mapping of the
covariates to the RKHS.

\subsection{Support Vector Machines}
\label{subsec:svm}

Support vector machines (SVMs) are a widely-used methodology for
two-class classification problems
\citep{cortes1995support,scholkopf2002learning}.  SVM aims to compute
a separating hyperplane of the form,
\begin{equation}
	\label{eq:hyperplane}
	f(\bX_i) \ = \ \bbeta^\top \phi(\bX_i)  + \beta_0,
\end{equation}
where $\bbeta \in \mathcal{H}_K$ is the normal vector for the
hyperplane and $\beta_0$ is the offset, and classification is done based
on which side of the hyperplane $\bX_i$ lies on.

In this paper, we use SVM for the classification of treatment status.
For non-separable data, $\bbeta$ and $\beta_0$ are computed according
to the following soft-margin SVM problem,
\begin{mini}[2]<b>
	{\bbeta,\beta_0,\bxi}{\frac{\lambda}{2}\norm*{\bbeta}_{\mathcal{H}_K}^2 + \sum_{i=1}^{N}\xi_i}{}{}
	\addConstraint{W_i f(\bX_i)}{\geq 1 - \xi_i,\quad}{i = 1,\dots, N \labelOP{eq:svm_primal}}
	\addConstraint{\xi_i}{\geq 0,}{i = 1,\dots, N,}
\end{mini}
where $\{\xi_i\}_{i=1}^N$ are the so-called slack variables, and
$\lambda$ is a regularization parameter controlling the trade-off
between the margin width of the hyperplane and margin violation of the
samples.  Note that $\lambda$ is related to the traditional SVM cost
parameter $C$ via $\lambda = 1/C$.

In cases where $\bbeta$ is high-dimensional, the dual form of
Eqn~\eqref{eq:svm_primal} is often preferable to work with. Defining
the matrix $\bQ$ with elements $\bQ_{ij} = W_iW_jK(\bX_i,\bX_j)$ and
the vector $\bW$ with elements $W_i$, the dual form is given by,
\begin{mini}[2]
	{\balpha}{\frac{1}{2\lambda} \balpha^\top\bQ\balpha - \bone^\top\balpha}{\label{eq:svm_dual}}{}
	\addConstraint{\bW^\top\balpha}{=0}{}
	\addConstraint{0 \preceq \balpha}{\preceq 1,}{}
\end{mini}
where $\balpha$ are called the dual coefficients, $\bone$ represents a vector of ones, 
and $\preceq$ denotes an element-wise inequality.

We begin by providing an intuitive explanation of how the SVM dual
given in Eqn~\eqref{eq:svm_dual} can be viewed as a covariate
balancing procedure.  First, note that the quadratic term in the above
dual objective function can be written as a weighted measure of
covariate discrepancy between the treatment and control groups,
\begin{equation}
	\label{eq:cov_discrep}
	\balpha^\top \bQ\balpha \ = \ \norm*{\sum_{i \in \mathcal{T}} \alpha_i \phi(\bX_i) - \sum_{i \in \mathcal{C}} \alpha_i \phi(\bX_i)}^2,
\end{equation}
while the constraint $\bW^\top \balpha = 0$ ensures that the sum of
weights is identical between the treatment and control groups,
\begin{equation*}
	\bW^\top \balpha = 0 \Longleftrightarrow \sum_{i \in \mathcal{T}} \alpha_i  = \sum_{i \in \mathcal{C}} \alpha_i.
\end{equation*}

Lastly, the second term in the objective, $\bone^\top \balpha$, is
proportional to the sum of weights for each treatment group, since the
above constraint implies
$\sum_{i \in \mathcal{T}} \alpha_i = \sum_{i \in \mathcal{C}} \alpha_i
= \bone^\top \balpha / 2$.  Thus, SVM simultaneously minimizes the
covariate discrepancy and maximizes the effective sample size, which
in turn leads to minimization of the weighted difference-in-means in
the transformed covariate space.  In addition, the weights are bounded
as represented by the constraint $0 \le \alpha_i \le 1$ for all $i$,
leading to stable causal effect estimation \citep{tan:10}.

The choice of kernel function $K(\bX_i, \bX_j)$ and its corresponding
feature map $\phi$ determine the type of covariate balance enforced by
SVM, as shown in Eqn~\eqref{eq:cov_discrep}. In this paper, we focus
on the linear, polynomial, and radial basis function (RBF)
kernels. The linear kernel $K(\bX_i, \bX_j) = \bX_i^\top \bX_j$
corresponds to a feature map $\phi(\bX_i) = \bX_i$, and hence the
quadratic term $\balpha^\top \bQ\balpha$ measures the discrepancy in
the original covariates. The general form for the degree $d$
polynomial kernel with scale parameter $c$ is
$K(\bX_i, \bX_j) = \left(\bX_i^\top \bX_j + c \right)^d$.
For example, the quadratic kernel with $d=2$ leads to a discrepancy
measure of the original covariates, their squares, and all pairwise
interactions.
The final kernel considered in this paper is the RBF kernel with scale
parameter $\gamma$:
$K(\bX_i, \bX_j) = \exp \left(- \gamma \norm{\bX_i - \bX_j}^2
\right)$.  This kernel can be viewed as a generalization of the
polynomial kernel in the limit $d \to \infty$.

In addition, 
the Karush–Kuhn–Tucker (KKT)
conditions for soft-margin SVM lead to the following useful
characterization for a solution $\balpha$:
\begin{align}
	\label{eq:kkt}
	W_i f(\bX_i) &= 1 \implies 0 \leq \alpha_i \leq 1\nonumber \\
	W_i f(\bX_i) &< 1 \implies \alpha_i = 1 \\
	W_i f(\bX_i) &> 1 \implies \alpha_i = 0 .\nonumber
\end{align}
The set of units that satisfy $W_i f(\bX_i) > 1$ are easy to classify
and receive zero weight. These units occur in regions of little or no
overlap and are most difficult to balance. The set of units that
satisfy $W_i f(\bX_i) = 1$ are referred to as marginal support
vectors, while those that meet $W_i f(\bX_i) < 1$ are the non-marginal
support vectors.  Collectively, these two sets correspond to the units
that the optimal hyperplane has the most difficulty classifying. As we
see in the next section, the support vectors comprise a balanced
subset.

In sum, the SVM dual problem finds the bounded weights that minimize
the covariate discrepancy between the treatment and control groups
while simultaneously maximizing effective sample size.  The
regularization parameter $\lambda$ controls which of these two
components receives more emphasis.  SVM chooses optimal weights such
that difficult-to-balance units are given zero weight.  Given this
intuition, we now establish more formally a connection between SVM,
covariate balancing, and causal effect estimation.


\subsection{SVM as a Maximum Mean Discrepancy Minimizer}
\label{subsec:mmd}

We now show that SVM minimizes the maximum mean discrepancy (MMD) of
covariate distribution between the treatment and control groups.  The
MMD is a commonly used measure of distance between probability
distributions \citep{gretton2007kernel} that was recently proposed as
a metric for balance assessment in causal inference
\citep{zhu2018kernel}.  Specifically, we show that the SVM dual
problem given in Eqn~\eqref{eq:svm_dual} can be viewed as a
regularized optimization problem for computing weights which minimize
the MMD.


The MMD, which is also called the kernel distance, is a measure of
distance between two probability distributions based on the difference
in mean function values for functions in the unit ball of a RKHS.  The
MMD has found use in several statistical applications
\citep{gretton2007kernel_b,gretton2012kernel,sriperumbudur2011mixture}. Given
the unit ball RKHS
$\mathcal{F}_K = \{f \in \mathcal{H}_K : \norm{f}_{\mathcal{H}_K} \leq
1 \}$ and two probability measures $F$ and $G$, the MMD is defined as
\begin{equation}
	\label{eq:mmd}
	\gamma_K(F, G) \coloneqq \sup_{f \in \mathcal{F}_K} ~\abs*{\int f dF - \int f dG}.
\end{equation}
An important property of the MMD is that when $K$ is a characteristic
kernel (e.g., the Gaussian radial basis function kernel and Laplace
kernel), then $\gamma_K(F, G) = 0$ if and only if $F = G$
\citep{sriperumbudur2010hilbert}.

The computation of $\gamma_K(F, G)$ requires the knowledge of both $F$
and $G$, which is typically unavailable. In practice, an estimate of
$\gamma_K(F, G)$ using the empirical distributions $\widehat{F}_m$ and
$\widehat{G}_n$ can be computed as
\begin{equation}
	\label{eq:mmd_emp}
	\gamma_K(\widehat{F}_m, \widehat{G}_n) = \norm*{\frac{1}{m} \sum_{i: \bX_i \sim F} \phi(\bX_i) - \frac{1}{n} \sum_{j: \bX_j \sim G} \phi(\bX_j)}_{\mathcal{H}_K},
\end{equation}
where $m$ and $n$ are the size of the samples drawn from $F$ and $G$,
respectively. The properties of this statistic are well-studied
\citep[see e.g.,][]{sriperumbudur2012empirical}.  In causal inference,
the empirical MMD can be used to assess balance between the treated
and control samples \citep{zhu2018kernel}.  This is done by setting
$F = P(\bX_i \mid T_i = 1)$ and $G = P(\bX_i \mid T_i = 0)$.  Then,
the quantity $\gamma_K(\widehat{F}_m, \widehat{G}_n)$ gives a measure
of independence between the treatment assignment $T_i$ and the
observed pre-treatment covariates $\bX_i$.

Eqn~\eqref{eq:mmd_emp} suggests a weighting procedure that balances
the covariate distributions between the treatment and control groups
by minimizing the empirical MMD.  We define a weighted variant of the
empirical MMD as
\begin{equation}
	\label{eq:mmd_weighted}
	\gamma_K(\widehat{F}_{\balpha}, \widehat{G}_{\balpha}) \ = \
	\norm*{\sum_{i \in \mathcal{T}} \alpha_i \phi(\bX_i) - \sum_{j
			\in \mathcal{C}} \alpha_j \phi(\bX_j)} \ = \ \sqrt{\balpha^\top\bQ\balpha},
\end{equation}
where $\widehat{F}_{\balpha}$ and $\widehat{G}_{\balpha}$ denote the
reweighted empirical distributions under weights $\balpha$, which are
restricted to the simplex set,
\begin{equation}
	\mathcal{A}_{\text{simplex}} \ = \ \l\{\balpha \in \R^N : 0 \preceq \balpha
	\preceq 1, ~\sum_{i \in \mathcal{T}} \alpha_i = \sum_{j \in
		\mathcal{C}} \alpha_j = 1\r\}.
\end{equation}
The optimization problem for minimizing the empirical MMD is therefore
formulated as,
\begin{mini}[2]
	{\balpha}{\sqrt{\balpha^\top \bQ \balpha}}{\label{eq:mmd_opt_orig}}{}
	\addConstraint{\balpha}{\in \cA_{\text{simplex}}.}{}
\end{mini}

Computing weights according to this problem is generally not
preferable due to the lack of regularization, which leads to
overfitting and sparse $\balpha$, resulting in many discarded
samples. The following theorem establishes that the SVM dual problem
can be viewed as a regularized version of the optimization problem in
Eqn~\eqref{eq:mmd_opt_orig}.
\begin{theorem}[SVM Dual Problem as Regularized MMD Minimization]
	\label{thm:svm_mmd_equiv} \spacingset{1}
	Let $\balpha_*(\lambda)$ denote the solution to the SVM dual
	problem under $\lambda$, defined in
	Eqn~\eqref{eq:svm_dual}.  Consider the normalized weights
	$\widetilde{\balpha}_*(\lambda) = 2\balpha_*(\lambda) /
	\bone^\top\balpha_*(\lambda)$ such that
	$\widetilde{\balpha}_*(\lambda) \in
	\cA_{\text{simplex}}$. Then,
	\begin{enumerate}
		\item[(i)] There exists $\lambda_*$ such that
		$\widetilde{\balpha}_*(\lambda_*)$ is a solution to the MMD
		minimization problem, defined in
		Eqn~\eqref{eq:mmd_opt_orig}.
		\item[(ii)] The quantity
		$\widetilde{\balpha}_*(\lambda)^\top \bQ
		\widetilde{\balpha}_*(\lambda)$ is a monotonically
		increasing function of $\lambda$.
	\end{enumerate} 
\end{theorem}
Proof is given in Appendix~\ref{app:thm:svm_mmd_equiv}.
Theorem~\ref{thm:svm_mmd_equiv} shows that the regularization
parameter $\lambda$ controls the trade-off between the covariate
imbalance, measured as the MMD, and the effective sample size,
measured as the sum of the support vector weights
$\bone^\top \balpha$.  Thus, a greater size of support vector set may
lead to a worse covariate balance within that set.



\subsection{SVM as a Relaxation of the Largest Balanced Subset Selection}

SVM can also be seen as a continuous relaxation of the quadratic
integer program (QIP) for computing the largest balanced subset.  We
modify the optimization problem in Eqn~\eqref{eq:svm_dual} by
replacing the continuous constraint $0 \preceq \balpha \preceq 1$ with
the binary integer constraint $\balpha \in \{0,1\}^N$. This problem,
which we refer to as SVM-QIP, is given by:
\begin{mini}[2]
  {\balpha}{\frac{1}{2\lambda}\balpha^\top\bQ\balpha - \bone^\top \balpha}{\label{eq:largest_subset}}{}
  \addConstraint{\bW^\top\balpha}{=0}{}
  \addConstraint{\balpha}{\in \{0,1\}^N,}{}
\end{mini}

Interpreting the variables $\alpha_i$ as indicators of whether or not
unit $i$ is selected into the optimal subset, we see that the
objective is a trade-off between subset balance in the projected
features (first term) and subset size (second term), where balance is
measured by a difference in sums.  However, the constraint
$\bW^\top \balpha = 0$ requires the optimal subset to have an equal
number of treated and control units, so balancing the feature sums
also implies balancing the feature means.

Thus, the SVM dual in Eqn~\eqref{eq:svm_dual} is a continuous
relaxation of the largest balanced subset problem represented by
SVM-QIP in Eqn~\eqref{eq:largest_subset}, with the set of support
vectors comprising an approximation to the largest balanced subset, as
these are the units for which $\alpha_i > 0$. The quality of the
approximation is difficult to characterize generally since differences
between the two computed subsets are influenced by a number of
factors, including separability in the data, the value of $\lambda$,
and the choice of kernel.  However, in our own experiments, we observe
significant overlap between the two solutions.  This result suggests
that SVM uses non-integer weights to augment the SVM-QIP solution
without compromising balance in the selected subset.  This highlights
the advantages of weighting over matching.

In Appendix~\ref{app:othermethods}, we further motivate SVM as a
balancing procedure by connecting it to two existing covariate
balancing methods, cardinality matching
\citep{zubizarreta2014matching} and stable balancing weights
\citep{zubizarreta2015stable}.  Specifically, we show that SVM can be
viewed as a weighting analog to cardinality matching.  In addition, a
variant of the SVM problem defined in Eqn~\eqref{eq:svm_primal},
called L2-SVM, forms an optimization problem that simultaneously
minimizes balance and weight dispersion.  SVM differs from these
methods primarily in the way that balance is enforced.

\subsection{Regularization Path as a Balance-Sample Size Frontier}
\label{subsec:path_alg}

Another important advantage to using SVM to perform covariate
balancing is the existence of path algorithms, which can efficiently
compute the set of solutions to Eqn~\eqref{eq:svm_dual} over different
values of $\lambda$. Since Theorem~\ref{thm:svm_mmd_equiv} establishes
that $\lambda$ controls the trade-off between the MMD and subset size,
the path algorithm for SVM can be viewed as the weighting analog to
the balance-sample size frontier \citep{king2017balance}.  Below, we
briefly discuss the path algorithm and explain how the path can be
interpreted as a balance-sample size frontier.

Path algorithms for SVM were first proposed by
\cite{hastie2004entire}, who showed that the weights $\balpha$ and
scaled intercept $\alpha_0 \coloneqq \lambda \beta_0$ are piecewise
linear in $\lambda$ and presented an algorithm for computing the
entire path of solutions with a comparable computation cost to finding
a single solution. However, their algorithm was prone to numerical
problems and would fail in the presence of singular submatrices of
$\bQ$. Recent work on SVM path algorithms has addressed these
issues. In our analysis, we use the path algorithm presented in
\cite{sentelle2016simple}, which we briefly describe here.

The regularization path for SVM is characterized
by a sequence of breakpoints, representing the values of $\lambda$ at
which either one of the support vectors on the margin $W_if(X_i) = 1$
exits the margin, or a non-marginal observation reaches the margin.
Between these breakpoints, the coefficients of the marginal support
vectors $\alpha_i$ change linearly in $\lambda$, while the
coefficients of all other observations stay fixed as $\lambda$ is
changed. Since the KKT conditions must be met for any solution
$\balpha$, we can form a linear system of equations to compute how each
$\alpha_i$ and $\alpha_0$ changes with respect to $\lambda$.

Based on this idea, beginning with an initial solution corresponding
to some large initial value of $\lambda$, the path algorithm first
computes how the current marginal support vectors change with respect
to $\lambda$. Given this quantity, the next breakpoint in the path is
computed by decreasing $\lambda$ until a marginal support vector exits
the margin, i.e, $\alpha_i = 0$ or $\alpha_i = 1$, or a non-marginal
observation enters the margin, i.e., $W_i f(X_i) = 1$. At this point,
the marginal support vector set is updated, and the changes in
$\alpha_i$ and $\alpha_0$, as well as the next breakpoint, are
computed. This procedure repeats until the terminal value of $\lambda$
is reached.

The initial solution in the SVM regularization path corresponds to the
solution at $\lambda_{\text{max}}$ such that for any
$\lambda > \lambda_{\text{max}}$, the minimizing weight vector
$\balpha$ does not change. We assume without a loss of generality that
$n_T \leq n_C$. Then initially, $\alpha_i = 1$ for all
$i \in \mathcal{T}$, and the remaining weights are computed according
to,
\begin{argmini}[2]<b>
	{\balpha}{\balpha^\top\bQ\balpha}{}{}
	\addConstraint{\sum_{i \in \mathcal{C}} \alpha_i}{= n_T\labelOP{eq:path_init}}{}
	\addConstraint{\alpha_i}{= 1,}{i \in \cT}
	\addConstraint{\balpha}{\in [0,1]^N.}
\end{argmini}
The initial solution computes control weights to minimize the weighted
empirical MMD while fixing the renormalized weights
$\widetilde{\alpha}_i = n_T^{-1}$ for the treated units. This
corresponds to the largest subset, as measured by $\sum_i \alpha_i$,
amongst all solutions on the regularization path.

The regularization path completes when the resulting solution has no
non-marginal support vectors, or in the case of non-separable data,
when $\lambda = 0$. In practice, however, path algorithms run into
numerical issues when $\lambda$ is small, so we terminate the path at
$\lambda_{\text{min}} = 1\times 10^{-3}$, which appears to work well
in our experiments. This value is often greater than $\lambda^*$,
which corresponds to the MMD-minimizing solution defined in
Theorem~\ref{thm:svm_mmd_equiv}.  In our experience, however, the
differences in balance between these two solutions are negligible.


Summarizing the regularization path, we see that the initial solution
at $\lambda_{\text{max}}$ has the largest weight sum
$\textstyle\sum_{i=1}^{N} \alpha_i = 2\min\{n_T,n_C\}$ and can be
viewed as the largest balanced subset retaining all observations in
the minority class. As we move through the path, the SVM dual problem
imposes greater restrictions on balance in the subset, which leads to
smaller subsets with better balance, until we reach the terminal value
$\lambda_{\text{min}}$, at which the weighted empirical MMD is
smallest on the path.

\subsection{Causal Effect Estimation}
\label{subsec:effect_est}

Theorem~\ref{thm:svm_mmd_equiv} establishes that the SVM dual problem
can be viewed as a regularized optimization problem for computing
balancing weights to minimize the MMD. However, achieving a high
degree of balance often requires significant pruning of the original
sample, especially in scenarios where the covariate distributions for
the treated and control groups have limited overlap. In this section,
we provide a characterization of this trade-off between subset size
and subset balance and discuss its impact on the bias of causal effect
estimates.

Recent work by \cite{kallus2020generalized} established that several
existing matching and weighting methods are minimizing the dual norm
of the bias for a weighted estimator, a property called error dual
norm minimizing. The author proposes a new method, kernel optimal
matching (KOM), which minimizes the dual norm of the bias when the
conditional expectation functions are embedded in an RKHS. The idea is
that the RKHS encapsulates a general class of functions that can
accommodate non-linearity and non-additivity in the conditional
expectation functions, leading to balancing weights which are robust
to model misspecification.  Below, we show that SVM also fits into
this KOM framework.

We restrict our attention to the following weighted
difference-in-means estimator,
\begin{equation}
	\label{eq:ate_est}
	\widehat{\tau} \ = \ \sum_{i \in \mathcal{T}} \alpha_i Y_i - \sum_{i \in \mathcal{C}} \alpha_i Y_i,
\end{equation}
where $\balpha \in \cA_{\text{simplex}}$ is computed via the
application of SVM to the data $\{\bX_i, T_i\}_{i=1}^N$.  Below, we
derive the form for the conditional bias with respect to two
estimands, the sample average treatment effect (SATE),
$\tau_{\text{SATE}}$, and the sample average treatment effect for the
treated (SATT), $\tau_{\text{SATT}}$.  We then discuss how to compute
this bias when the conditional expectation functions $f_0$ and $f_1$
are unknown.

As shown in Appendix~\ref{app:bias}, under
Assumptions~\ref{assn:uncon}~and~\ref{assn:over}, the conditional bias
with respect to $\tau_{\text{SATE}}$ and $\tau_{\text{SATT}}$ for the
estimator above is given by,
\begin{equation}
	\label{eq:bias_cate}
	\E(\widehat{\tau} - \tau \mid \{\bX_i, T_i\}_{i=1}^N) \ = \ \sum_{i=1}^{N} \alpha_i W_i f_{0}(\bX_i) + \sum_{i=1}^N \left(\alpha_iT_i - v_i\right) \tau(\bX_i),
\end{equation}
where
\begin{equation}
	v_i \ = \ 
	\begin{cases}
		1/N &\text{if } \tau = \tau_{\text{SATE}} \\
		T_i /n_T &\text{if } \tau = \tau_{\text{SATT}}
	\end{cases},
\end{equation}      
and
$\tau(\bX_i) \coloneqq \E(Y_i(1) - Y_i(0) \mid \bX_i) = f_1(\bX_i) -
f_0(\bX_i)$. 

The first term in Eqn~\eqref{eq:bias_cate} represents the bias due to
the imbalance of prognostic score \citep{hansen2008prognostic}.
Unfortunately, computation of this quantity is difficult since $f_0$
is typically unknown.  However, we can embed $f_0$ in a unit-ball
RKHS, $\mathcal{F}_K$, and consider an $f_0$ that maximizes the square
of this bias term.  As shown in Appendix~\ref{app:worst-case}, this
strategy leads to the following optimization problem that is of the
same form as that of the MMD minimization problem given in
Eqn~\eqref{eq:mmd_opt_orig},
\begin{mini}
	{\balpha}{\gamma_{K}^2\left(\widehat{F}_{\balpha}, \widehat{G}_{\balpha} \right)}{\label{eq:wc_bias_conf}}{}
	\addConstraint{\balpha}{\in \cA_{simplex}.}{}
\end{mini}
Thus, the SVM dual problem can also be viewed as a method for minimizing bias due to
prognostic imbalance.

However, SVM does not address the second term of the conditional bias
in Eqn~\eqref{eq:bias_cate}.  This term represents the bias due to
extrapolation outside of the weighted treatment group to the
population of interest.
For example, if the SATT is the target quantity, the term represents
the bias due to the difference between the weighted and unweighted
CATE for the treatment group.  Thus, the prognostic balance achieved
by SVM may induce this CATE bias.  This is a direct consequence of the
trade-off between balance and effective sample size, which is
controlled by the regularization parameter as shown earlier.


\subsection{Relation to Kernel Balancing Methods}
\label{subsec:kbal}

We now discuss the relations between SVM and a class of closely
related methods, known as kernel balancing \citep{wong2018kernel,
  hazlett2018kernel,kallus2021more}.  Consider the following
alternative decomposition of conditional bias,
\begin{equation}
	\label{eq:bias_general}
	\E(\widehat{\tau} - \tau \mid \{\bX_i, T_i\}_{i=1}^N) \ = \ \sum_{i=1}^N
	\left(\alpha_iT_i - v_i\right) f_1(\bX_i) + \sum_{i=1}^{N} \left\{v_i
	- \alpha_i(1-T_i)\right\} f_{0}(\bX_i).
\end{equation}
To minimize this bias, kernel balancing methods restrict $f_0$ and
$f_1$ to a RKHS
$\mathcal{F}_k \coloneqq \{(f_0, f_1) \in \mathcal{H}_{K_0} \times \mathcal{H}_{K_1} : \sqrt{\norm{f_0}_{\mathcal{H}_{K_0}}^2 + \norm{f_1}_{\mathcal{H}_{K_1}}^2} \leq c\}$
for some $c$ and consider minimizing the largest bias under the pair
$(f_0, f_1) \in \mathcal{F}_K$.  This problem is given by,
\begin{mini}
	{\balpha}{\sup_{f_0,f_1 \in \mathcal{F}_K} ~~ \left[\sum_{i=1}^N \left(\alpha_iT_i - v_i\right) f_1(\bX_i) - \sum_{i=1}^{N} \left\{\alpha_i(1-T_i) - v_i\right\} f_{0}(\bX_i) \right]^2}{\label{eq:wc_bias_min}}{}
	\addConstraint{\balpha}{\in \cA,}{}
\end{mini} 
where $\cA$ denotes the constraints on the weights.  For example,
\citet{kallus2021more} restricts them to $\cA_{\text{simplex}}$
whereas \citet{wong2018kernel} essentially uses
$\alpha_i \geq N^{-1}$, though their formulation is slightly different
than that given above.  Note that approaches targeting ATT for
estimation, such as the one considered by \cite{hazlett2018kernel},
fix all treated weights to $\alpha_i = n_T^{-1}$ and focus only on the
term involving $f_0(\bX_i)$.

As shown in~\citet{kallus2021more}, the problem of minimizing this
worst-case conditional bias amounts to computing weights that balance
the treatment and control covariate distributions with respect to the
empirical distribution for the population of interest.  Let
$\widehat{F}_{\balpha}$ and $\widehat{G}_{\balpha}$ denote the
weighted empirical covariate distributions for the treatment and
control groups, respectively, while having $\widehat{H}_{\bv}$
represent the empirical covariate distribution corresponding to the
population of interest.  Then if $\mathcal{F}_K$ is also restricted to
the unit-ball RKHS (fixing the size of $f_0,f_1$ is necessary since
the bias scales linearly with $\norm{f_0}_{\mathcal{H}_{K_0}}$ and
$\norm{f_1}_{\mathcal{H}_{K_1}}$), i.e., $c=1$, the optimization
problem in Eqn~\eqref{eq:wc_bias_min} can be written in terms of the
minimization of the empirical MMD statistic:
\begin{mini}
	{\balpha}{\gamma_{K_1}^2\left(\widehat{F}_{\balpha}, \widehat{H}_{\bv} \right) + \gamma_{K_0}^2\left(\widehat{G}_{\balpha}, \widehat{H}_{\bv}\right)}{\label{eq:wc_bias_mmd}}{}
	\addConstraint{\balpha}{\in \cA.}{}
\end{mini}
This objective does not contain a measure of distance between the
conditional covariate distributions $\widehat{F}_{\balpha}$ and
$\widehat{G}_{\balpha}$. Instead, balance between these two
distributions is indirectly encouraged through balancing each one
individually with respect to the target distribution
$\widehat{H}_{\bv}$.  This is in contrast with SVM, which directly
balances the covariate distribution between the treatment and control
groups.

\section{Simulations}
\label{sec:sims}

In this section, we examine the performance of SVM in ATE estimation
under two different simulation settings.  We also examine the
connection between SVM and the QIP for the largest balanced subset.

\subsection{Setup}
\label{subsec:setup}

We consider two simulation setups used in previous studies.
Simulation~A comes from \cite{lee2010improving} who use a slightly
modified version of the simulations presented in
\cite{setoguchi2008evaluating}.  We adopt the exact setup
corresponding to their ``scenario G,'' which is briefly summarized
here.  We refer readers to the original article for the exact
specification.  For each simulated dataset, we generate 10 covariates
$\bX_i = (X_{i1}, \dots, X_{i10})^\top$ from the standard normal
distribution, with correlation introduced between four pairs of
variables. Treatment assignment is generated according to
$P(T_i = 1 \mid \bX_i) = \text{expit}\left(\bbeta^\top
  f(\bX_i)\right)$, where $\bbeta$ is some coefficient vector, and
$f(\bX_i)$ controls the degree of additivity and linearity in the true
propensity score model. This scenario uses the true propensity score
model with a moderate amount of non-linearity and non-additivity. The
outcome model was specified to be linear in the observed covariates
with a constant, additive treatment effect:
$Y_i(T_i) = \gamma_0 + \bgamma^\top \bX_i + \tau T_i + \epsilon_i$,
with $\tau = -0.4$ and $\epsilon_i \sim \mathcal{N}(0,0.1)$.

Simulation~B comes from \cite{wong2018kernel} and represents a more
difficult scenario where both the propensity score and outcome
regression models are nonlinear in the observed covariates. For each
simulated data set, we generate a ten-dimensional random vector
$\bZ_i = (Z_{i1}, \dots, Z_{i10})^\top$ from the standard normal
distribution. The observed covariates are the nonlinear functions of
these variables, $\bX_i = (X_{i1}, \dots, X_{i10})^\top$, where
$X_1 = \exp(Z_1/2)$, $X_2 = Z_2 / [1+\exp(Z_1)]$,
$X_3 = (Z_1Z_3/25 + 0.6)^3$, $X_4 = (Z_2 + Z_4 + 20)^2$, and
$X_j = Z_j,\, j=5, \dots, 10$. Treatment assignment follows
$P(T_i = 1 \mid \bZ_i) = \text{expit}(-Z_1 - 0.1Z_4)$, which
corresponds to Model~1 of \cite{wong2018kernel}.  Finally, the outcome
model is specified as
$Y(T_i) = 200 + 10T_i + (1.5T_i - 0.5)(27.4Z_1 + 13.7Z_2 + 13.7Z_3 +
13.7Z_4) + \epsilon_i$, with $\epsilon_i \sim \mathcal{N}(0,1)$. Note
that the true PATE is $10$ under this model.

\subsection{Comparison between SVM and SVM-QIP}
\label{subsec:valid}

We begin by examining the connection between SVM and SVM-QIP by
comparing solutions obtained using one simulated dataset of $N=500$
units under Simulation~A. Specifically, we first compute the SVM path
using the path algorithm described in Section~\ref{subsec:path_alg},
obtaining a set of regularization parameter breakpoints
$\lambda$. Next, we compute the SVM-QIP solution for each of these
breakpoints using the Gurobi optimization software \citep{gurobi}. We
limit the solver to spending 5 minutes of runtime for each problem.
Finding the \emph{exact} integer-valued solution under a given
$\lambda$ requires a significant amount of time, but a good
approximation can typically be found in a few seconds.

For both methods, we compute the objective function value at each of
the breakpoints as well as the coverage of the SVM-QIP solution by the
SVM solution.  The latter represents the proportion of units with
non-zero SVM weights that are included in the largest balanced subset
identified by SVM-QIP.  Formally, the coverage is defined as
\begin{equation*}
	\text{cvg}(\lambda) \ = \ \frac{\left|\lceil\balpha_{\text{SVM}}(\lambda)\rceil \cap \balpha_{\text{SVM-QIP}}(\lambda)\right|}{|\balpha_{\text{SVM-QIP}}(\lambda)|}.
\end{equation*}
To examine the effects of separability on the quality of the
approximation, we perform the above analysis using three different
types of features. Specifically, we use a linear kernel with the
untransformed covariates (linear), a linear kernel with the degree-2
polynomial features formed by concatenating the original covariates
with all two-way interactions and squared terms (polynomial), and the
Gaussian RBF with scale parameter chosen according the median
heuristic (RBF). In all cases, we scale the input feature matrix such
that the columns have $0$ mean and standard deviation $1$ before
performing the kernel computation.

\begin{figure}[t!]
	\centering
	\includegraphics[width=\textwidth]{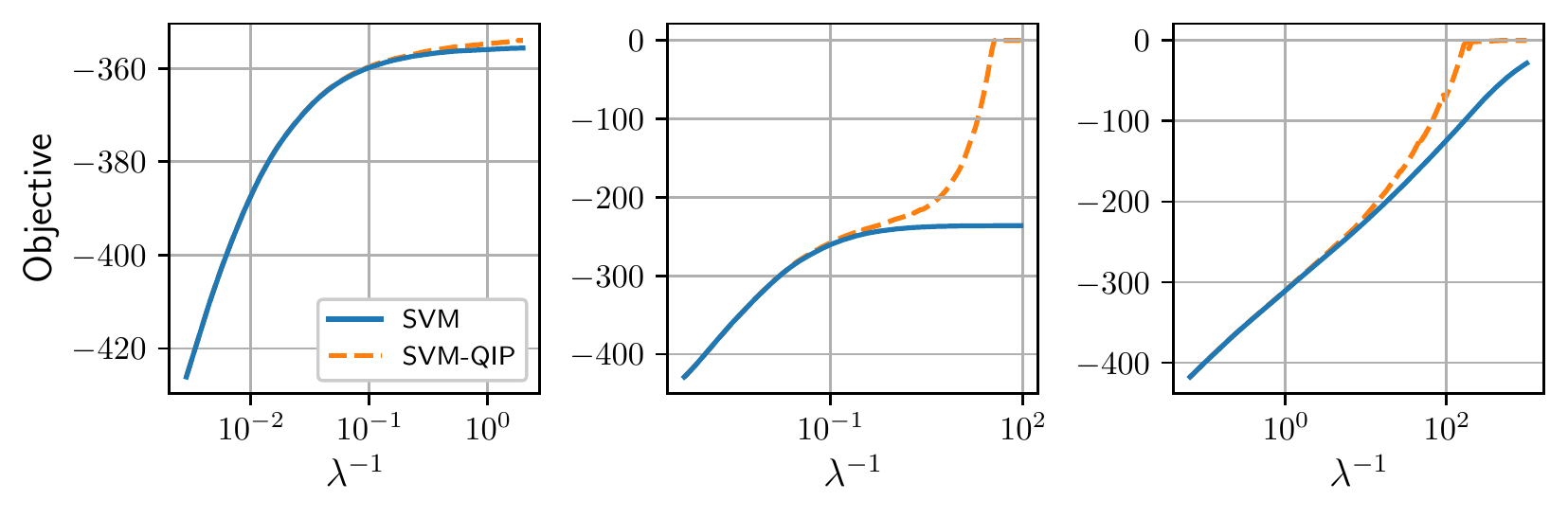}
	\begin{flushleft}
		\vspace{-0.5cm}
		\ifprint
		\else
		\begin{tabular}{C{1.285cm}  C{3.33cm}  C{0.85cm}  C{3.33cm}  C{0.95cm}   C{3.33cm} }
			& \small(a) Linear & & \small (b) Polynomial & & \small (c) RBF
		\end{tabular}
		\fi
		\vspace{-0.4cm}
	\end{flushleft}
	\caption[Comparison of the objective value between SVM and SVM-QIP]{
		Comparison of the objective value between SVM and
		SVM-QIP. The blue line denotes the objective value of the SVM
		solution, and the orange dotted line denotes the objective value of the
		SVM-QIP solution.}
	\label{fig:val_obj}
\end{figure}

Figure~\ref{fig:val_obj} shows that the objective values for the SVM
and SVM-QIP solutions are close when the penalty on balance
$\lambda^{-1}$ is small, with divergence between the two methods
occurring towards the end of the regularization path. In the linear
case, we see that the paths for the two methods are nearly identical,
suggesting that their solutions are essentially the same. Divergence
in the polynomial and RBF settings is more pronounced due to greater
separability in the transformed covariate space, which is more
difficult to balance without non-integer weights. When $\lambda$ is
very small, we also find that SVM-QIP returns $\balpha = \bzero$,
indicating that the penalty on balance is too great. Lastly, the
effects of approximating the SVM-QIP solution are reflected in the RBF
setting, where upon close inspection the objective value appears to be
somewhat noisy and non-monotonic.

\begin{figure}[t!]
	\centering
	\includegraphics[width=\textwidth]{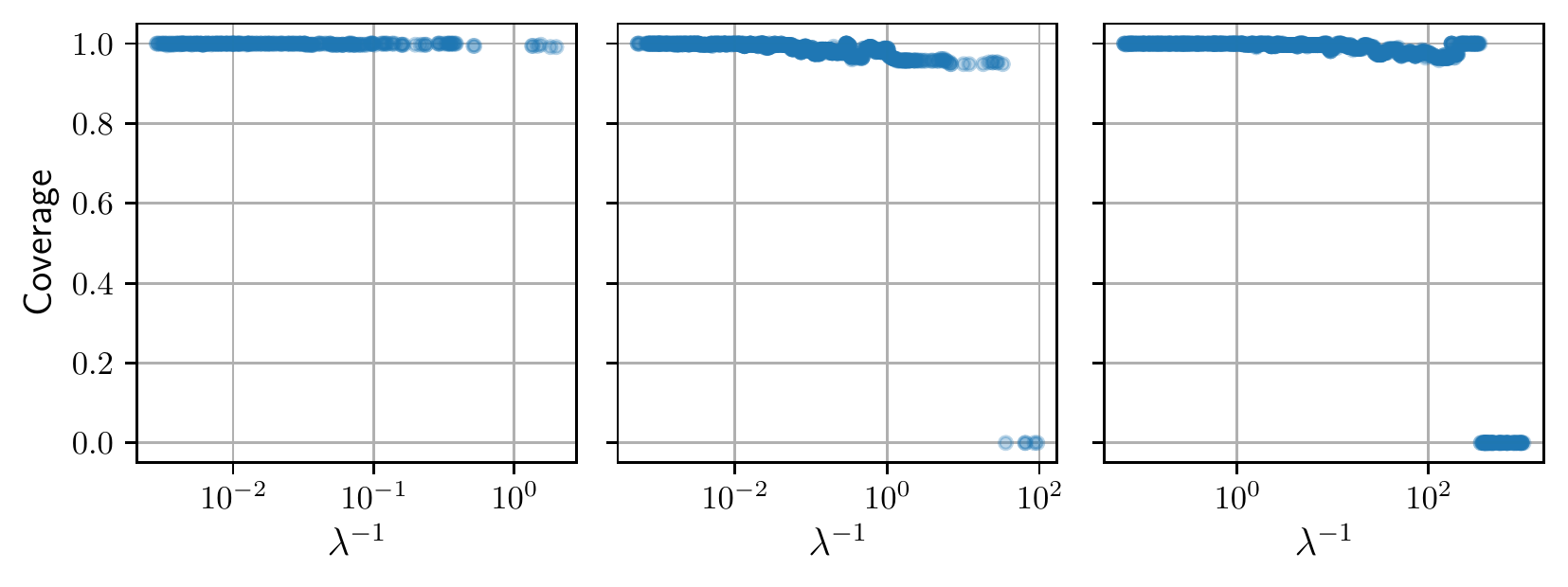}
	\begin{flushleft}
		\vspace{-0.5cm}
		\ifprint
		\else
		\begin{tabular}{C{0.95cm}  C{4.025cm}  C{4.95cm}  C{4.025cm} }
			& \small(a) Linear & \small (b) Polynomial & \small (c) RBF
		\end{tabular}
		\fi
		\vspace{-0.4cm}
	\end{flushleft}
	\caption[Proportion of samples in the SVM-QIP largest
	balanced subset covered by the SVM solution]{
		The proportion of samples in the SVM-QIP largest
		balanced subset covered by the SVM solution.  The instances
		of zero coverage in the polynomial and RBF settings
		represent the cases where the SVM-QIP fails to find a
		nontrivial solution.}
	\label{fig:val_cov}
\end{figure}

Interestingly, the coverage plots in Figure~\ref{fig:val_cov} show
that even when the objective values of the two methods are divergent,
the SVM solution still predominantly covers the SVM-QIP solution.  The
regions with zero coverage in the polynomial and RBF settings
correspond to instances where the balance penalty is so significant
that a nontrivial solution cannot be found for the SVM-QIP. This
result illustrates that SVM approximates one-to-one matching by
augmenting a well-balanced matched subsample with some non-integer
weights.  This leads to an increased subset size while preserving the
overall balance within the subsample.

\subsection{Performance of SVM}
\label{subsec:sim_results}

Next, we evaluate the performance of SVM in estimating the ATE for
Simulations~A~and~B. For each scenario, we generate 1,000 datasets
with $N=500$ samples.  For each simulated dataset, we compute the ATE
estimate over a fixed grid of $100$ $\lambda$ values chosen based on
the simulation scenario and input feature.  As described in
Section~\ref{subsec:valid}, we use the linear, polynomial, and
RBF-induced features, standardizing the covariate matrix before
passing it to the kernel in all cases.

\begin{figure}[t!]
	\centering
	\includegraphics[width=\textwidth]{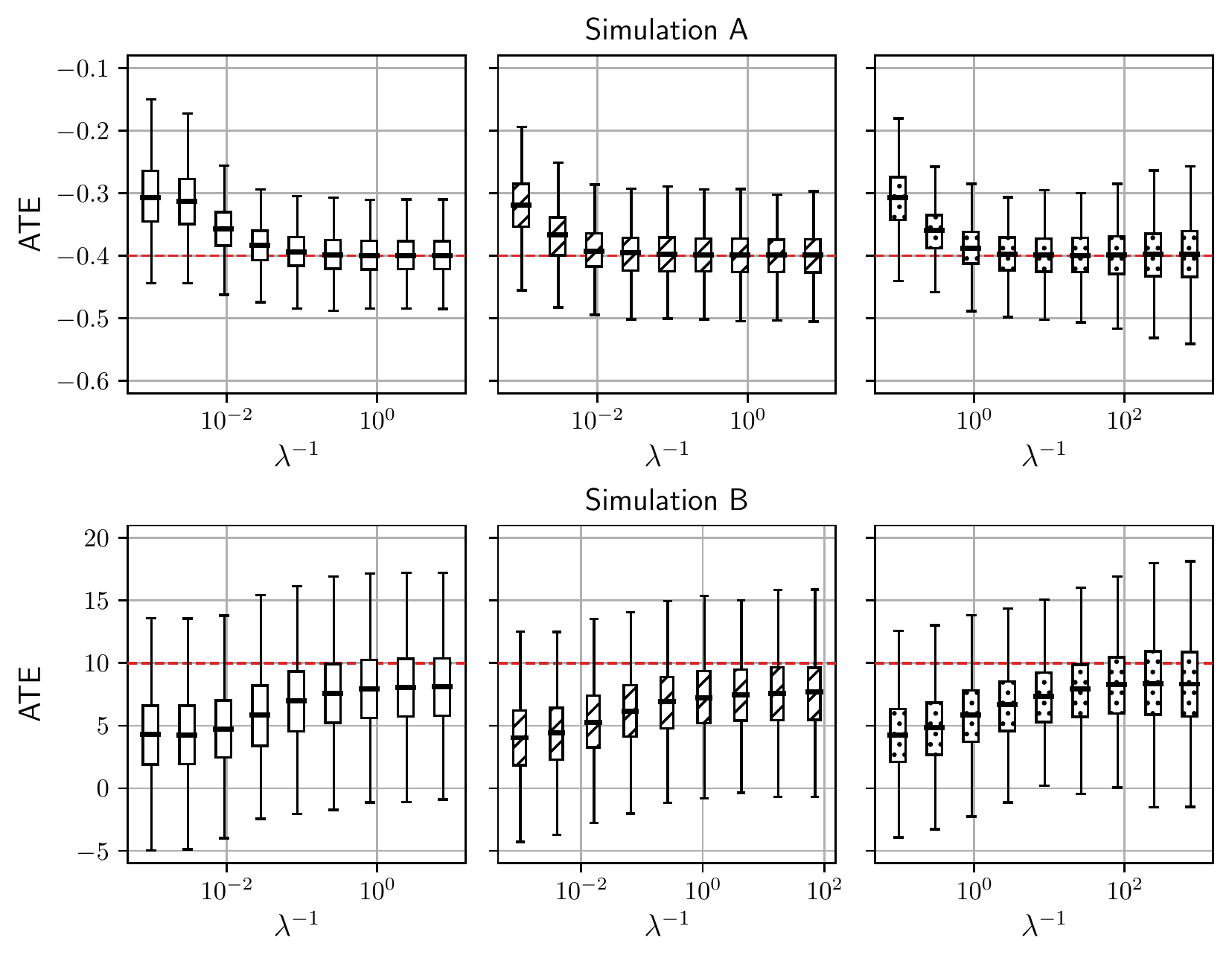}
	\begin{flushleft}
		\vspace{-0.5cm}
		\ifprint
		\else
		\begin{tabular}{C{1.2cm}  C{3.95cm}  C{4.85cm}  C{3.95cm} }
			& \small(a) Linear & \small (b) Polynomial & \small (c) RBF
		\end{tabular}
		\fi
		\vspace{-0.4cm}
	\end{flushleft}
	\caption[ATE estimates for Simulations~A and B over the SVM regularization path]{
		ATE estimates for Simulations~A (top) and B (bottom)
		over the SVM regularization path. The boxplots represent the
		distribution of the ATE estimates over Monte Carlo
		simulations. The red dashed line corresponds to the true
		ATE.}
	\label{fig:path_sims}
\end{figure}

Figure~\ref{fig:path_sims} plots the distribution of ATE estimates
over Monte Carlo simulations against the regularization parameter
$\lambda$.  The results for Simulation~A (top panel) show that the
bias approaches zero as the penalty on balance increases ($\lambda$
decreases).  This is because the conditional bias in the estimate
under the outcome model for Simulation~A is given by,
\begin{equation*}
	\E[\widehat{\tau} - \tau \mid X_{1:N}, T_{1:N}] \ = \ \bgamma^\top \left(\sum_{i=1}^{N}\alpha_i T_i \bX_i\right).
\end{equation*}
This implies that all bias comes from prognostic score imbalance.
This quantity becomes the smallest when minimizing
$\norm{\sum_{i=1}^{N}\alpha_i T_i \bX_i}$, which is controlled by the
regularization parameter $\lambda$ in the SVM dual objective under the
linear setting.  Note that this quantity is also small under 
both the polynomial and RBF input features.

We also find that under all three settings, there is relatively little
change in the variance of the estimates along most of the path,
suggesting that the variance gained from trimming the sample is
counteracted by the variance decreased from correcting for
heteroscedasticity. The exception to this observation occurs at the
beginning of the linear case, where the reduction in bias also reduces
the variance, and at the end of the RBF path, where the amount of
trimming is so substantial relative to the balance gained that the
variance increases.

For Simulation~B (bottom panel), the bias decreases as the penalty on
balance increases.  However, due to misspecification, nonlinearity,
and treatment effect heterogeneity in the outcome model, the bias
never decays to zero as shown in Section~\ref{subsec:effect_est}. We
also find that the SVM with linear kernel can reduce bias as well as
the other kernels, suggesting that SVM is robust to misspecification
and nonlinearity in the outcome model.  Similar to Simulation~A, we
observe relatively small changes in the variance as the constraint on
balance increases, except at the end of the RBF path where there is
substantial sample pruning.

\subsection{Comparison with Other Methods}

Next, we compare the performance of SVM with that of other methods.
Our results below show that the performance of SVM is comparable to
that of related state-of-the-art covariate balancing methods available
in the literature.  In particular, we consider kernel optimal matching
\citep[KOM;][]{kallus2021more}, kernel covariate balancing
\citep[KCB;][]{wong2018kernel}, cardinality matching
\citep[CARD;][]{zubizarreta2014matching}, and inverse propensity score
weighting (IPW) based on logistic regression (GLM) and random forest
(RFRST), both of which were used in the original simulation study by
\cite{lee2010improving}. For SVM, we compute solutions using
$\lambda^{-1} = 0.42$, $\lambda^{-1} = 0.10$, and
$\lambda^{-1} = 2.60$ for Simulation A under the linear, polynomial,
and RBF settings, respectively.  For Simulation B, we use
$\lambda^{-1} = 1.07$, $\lambda^{-1} = 1.92$, and
$\lambda^{-1} = 10.48$. These values are taken from the grid of
$\lambda$ values used in the simulation based on visual inspection of
the path plots in Figure~\ref{fig:path_sims} around where the estimate
curve flattens out.

For KOM, we compute weights under the linear, polynomial, and RBF
settings described earlier with the default settings for the provided
code. For KCB, we compute weights using the RBF kernel and use its
default settings.  While KCB allows for other kernel functions, it was
originally designed for the use of RBF and Sobolev kernels.  We find
its results to be poor when using the linear and polynomial features.
For CARD, we used a threshold of $0.01$ and $0.1$ times the
standardized difference-in-means for the linear- and
polynomial-induced features, respectively, and we set the search time
for the algorithm to 3 minutes. For GLM and RFRST, we used the
linear-induced features and the default algorithm settings described
in \cite{lee2010improving}.



\begin{figure}[t!]
	\includegraphics[width=\textwidth]{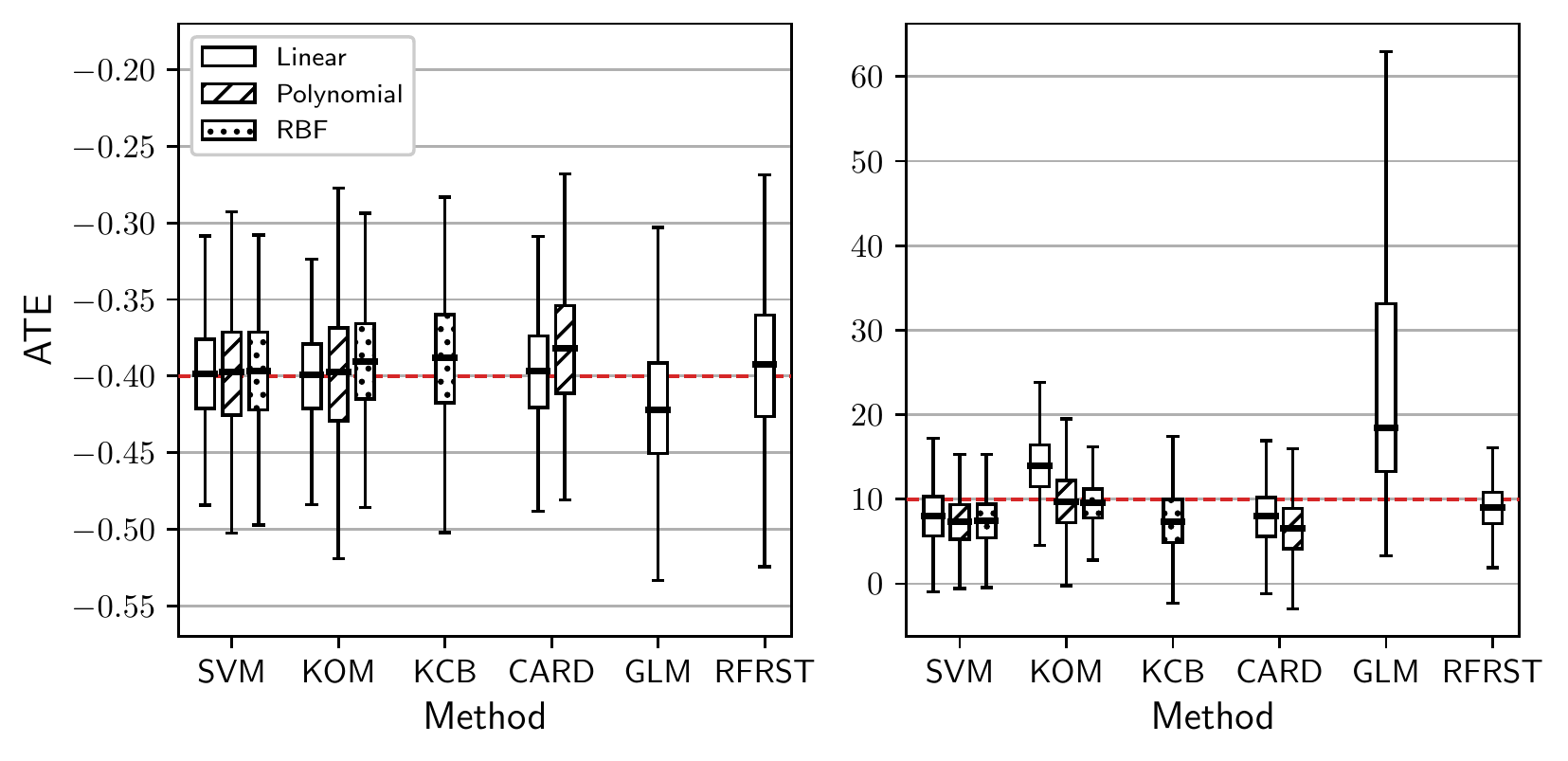}
	\begin{flushleft}
		\vspace{-0.5cm}
		\ifprint
		\else
		\begin{tabular}{C{1.39cm}  C{5.8cm}  C{0.7cm}  C{5.8cm} }
			& \small (a) Simulation A & & \small (b) Simulation B
		\end{tabular}
		\fi
		\vspace{-0.4cm}
	\end{flushleft}
	\caption[Boxplots for ATE estimates for
	Simulations~A and B]{
		Boxplots for ATE estimates for
		Simulations~A~(left)~and~B (right).  The hatch pattern
		denotes the input feature (Linear, Polynomial, or RBF) ---
		kernel optimal matching \citep[KOM;][]{kallus2021more},
		kernel covariate balancing \citep[KCB;][]{wong2018kernel},
		cardinality matching
		\citep[CARD;][]{zubizarreta2014matching}, and IPW with
		propensity score modeling via logistic regression (GLM) and random
		forest (RFRST). The red dashed line corresponds to the true
		ATE.}
	\label{fig:compare_sims}
\end{figure}

Figure~\ref{fig:compare_sims} plots the distributions of the effect
estimates over 1,000 simulated datasets for both scenarios.
Simulation~A (left panel) shows comparable performance across all
methods, with SVM and KOM having the best performance in terms of both
bias and variance. In particular, SVM achieves near zero bias under
all three input features. The results for KCB show that it performs
slightly worse in comparison to the other kernel methods, with greater
bias and variance under the RBF setting.

The results for CARD show near identical performance with SVM under
the linear setting, however results under the polynomial setting are
notably worse. The reason for this comes from the choice of balance
threshold, which was set to 0.1 times the standardized
difference-in-means of the input feature matrix. Although decreasing
the scalar below 0.1 would lead to a more balanced matching, we found
that algorithm was unable to consistently find a solution for all
datasets with scalar multiples smaller than 0.1. This result
highlights the main issue with defining balance
dimension-by-dimension, which makes it difficult to enforce small
overall balance without information on the underlying geometry of the
data. Lastly, the propensity score methods show the worst
performance. This is somewhat expected as the true propensity score
model is more complicated than the true outcome model under this
simulation setting.

We note that further reduction in the variance of the SVM solution
while preserving bias is likely possible with a more principled method
of choosing the solution for each simulated dataset. In general, a
value of $\lambda$ that works well for one dataset may not work for
another.  A better approach would examine estimates over the path and
balance-sample size curves for each dataset
individually. Nevertheless, our heuristic procedure to selecting a
solution produced high-quality results.

The results for Simulation~B (right panel) show a slightly more
varying performance across methods. Amongst the kernel methods, we
find that KOM has the best performance under the polynomial and RBF
settings, achieving near zero bias under these scenarios, while SVM
has the best performance under the linear setting. The discrepancy
under the linear setting is due to misspecification, which leads to a
poor regularization parameter choice and consequently poor balance and
bias under the KOM procedure.

We also find that SVM is unable to drive the bias to zero, which is
due to the treatment effect heterogeneity in the outcome model. As
discussed in Section~\ref{subsec:effect_est}, SVM ignores the second
term in the conditional bias decomposition in
Eqn~\eqref{eq:bias_cate}, which is zero under a constant additive
treatment effect in Simulation A but is nonzero in Simulation B. In
contrast, KOM targets both bias terms in its formulation, which leads
to greater bias reduction.

In comparison to the other kernel methods, KCB has comparable bias to
SVM but greater variance. For CARD, we observe comparable results to
SVM under the linear setting, but worse performance under the
polynomial setting due to the reasons mentioned above. Lastly, we find
mixed results between the two propensity score methods. Logistic
regression (GLM) has the worst performance while Random forest (RFRST)
exhibits the second best performance. This result is likely due to the
simple structure of the true propensity score model, whose
nonlinearity can only be accurately modeled by RFRST.

\section{Empirical Application: Right Heart Catheterization Study}
\label{sec:rhc}

In this section, we apply the proposed methodology to the right heart
catheterization (RHC) data set originally analyzed in
\cite{connors1996effectiveness}. This observational data set was used
to study the effectiveness of right heart catheterization, a
diagnostic procedure, for critically ill patients. The key result from
the study was that after adjusting for a large number of pre-treatment
covariates, right heart catheterization appeared to reduce survival
rates.  This finding contradicts the existing medical perception that
the procedure is beneficial.

\subsection{Data and Methods}
\label{subsec:rhc_setup}

The data set consists of 5,735 patients, with 2,184 of them assigned
to the treatment group and 3,551 assigned to the control group. For
each patient, we observe the treatment status, which indicates whether
or not he/she received catheterization within 24 hours of hospital
admission.  The outcome variable represents death within 30 days.
Finally, the dataset contains a total of 72 pre-treatment covariates
that are thought to be related to the decision to perform right heart
catheterization. These variables include background information about the
patient, such as age, sex, and race, indicator variables for primary/secondary
diseases and comorbidities, and various measurements from medical test
results.  

We compute the full SVM regularization paths under the linear,
polynomial, and RBF settings described in
Section~\ref{subsec:setup}. For the polynomial features, we exclude
all trivial interactions (e.g., interactions between categories of the
same categorical variable) and squares of binary-valued covariates. We
also compute the KOM weights under all three settings, the KCB weights
under the RBF setting, and the CARD weights under the linear and
polynomial settings with a threshold set to 0.1 times the standardized
difference-in-means.

\subsection{Results}
\label{subsec:rhc_results}

\begin{figure}[t!]
	\centering
	\includegraphics[width=\textwidth]{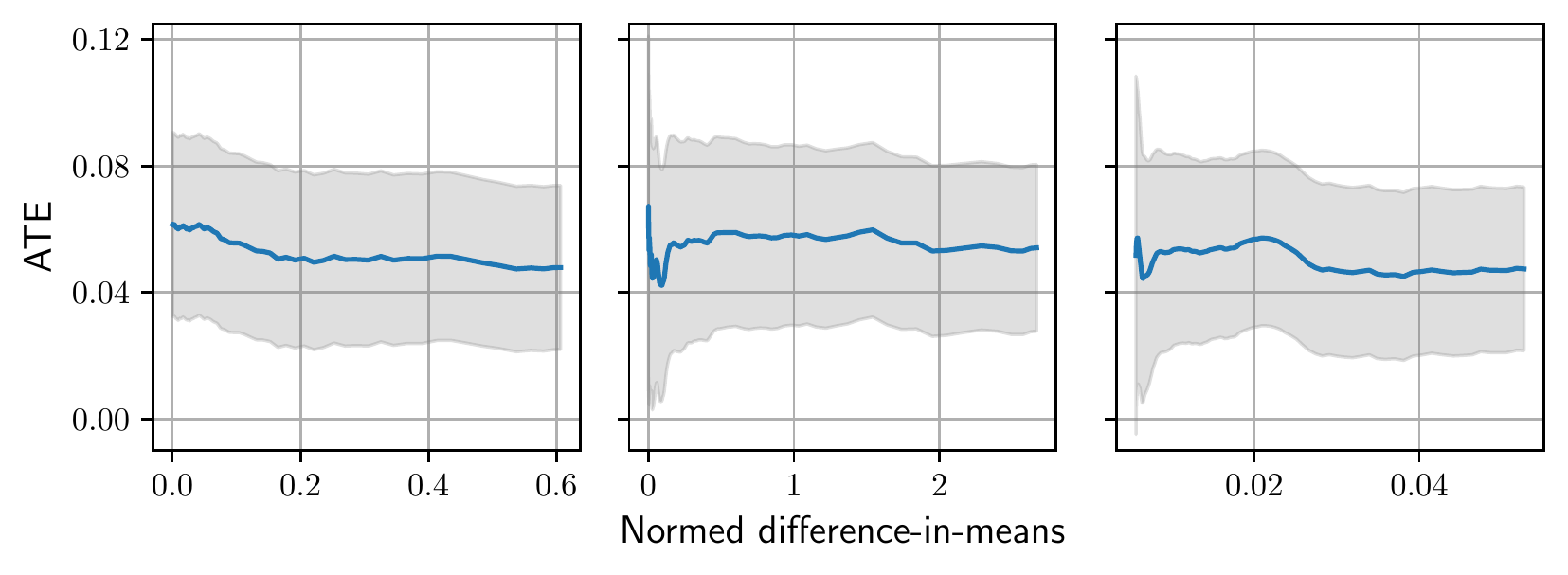}
	\begin{flushleft}
		\vspace{-0.5cm}
		\ifprint
		\else
		\begin{tabular}{C{1.12cm}  C{3.9cm}  C{5cm}  C{3.9cm} }
			& \small(a) Linear & \small (b) Polynomial & \small (c) RBF
		\end{tabular}
		\fi
		\vspace{-0.4cm}
	\end{flushleft}
	\caption[ATE estimates for the RHC data over the SVM
	regularization path]
	{ATE estimates for the RHC data over the SVM
		regularization path. The horizontal axis represents the
		normed difference-in-means in covariates within the weighted
		subset. The solid blue line denotes the average estimate,
		and the solid gray background denotes the pointwise 95\%
		confidence intervals.}
	\label{fig:path_rhc}
\end{figure}

Figure~\ref{fig:path_rhc} plots the ATE estimates over the SVM
regularization paths with the pointwise 95\% confidence intervals
based on the weighted Neyman variance estimator \citep[Chapter~19]{imbens2015causal}.  The horizontal axis
represents the normed difference-in-means within the weighted subset
as a covariate balance measure.  For all three settings, we find that
the estimated ATE slightly increases as the weighted subset becomes
more balanced, supporting the results originally reported in
\cite{connors1996effectiveness} that right heart catheterization
decreased survival rates.

\begin{figure}[t!]
	\centering
	\includegraphics[width=\textwidth]{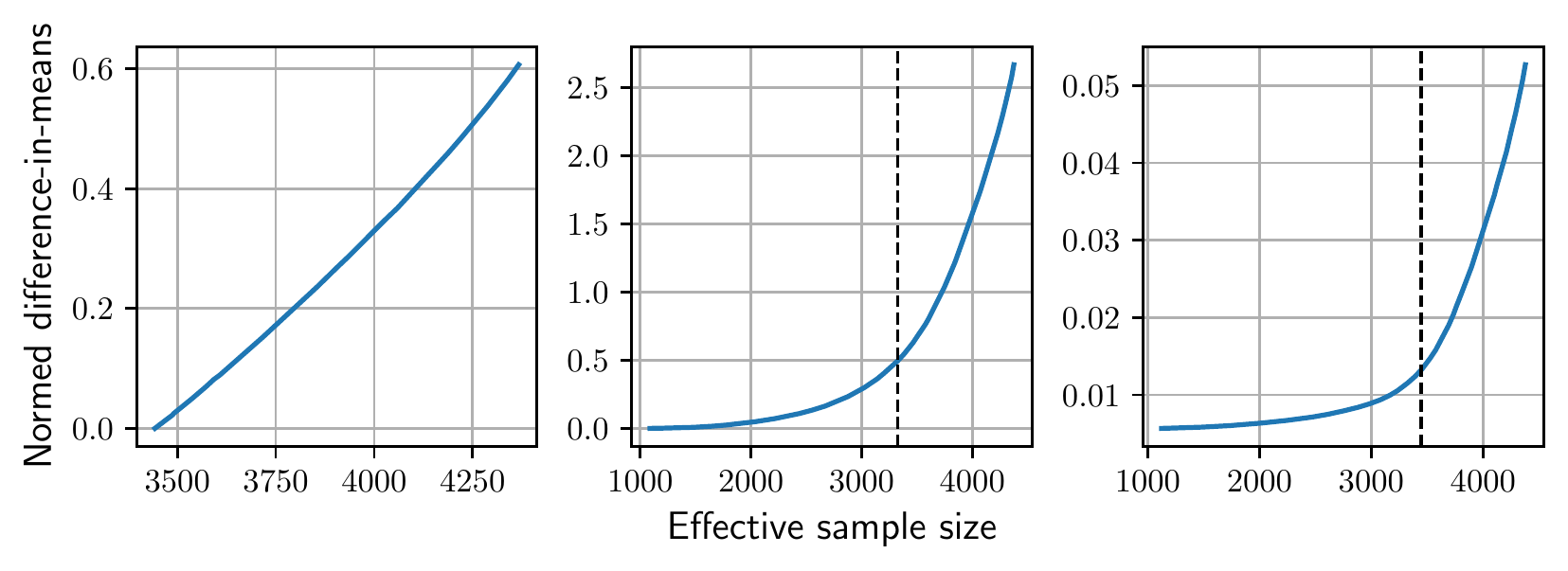}
	\begin{flushleft}
		\vspace{-0.5cm}
		\ifprint
		\else
		\begin{tabular}{C{0.97cm}  C{3.65cm}  C{0.5cm}  C{3.65cm}  C{0.65cm}  C{3.65cm} }
			& \small(a) Linear & & \small (b) Polynomial & & \small (c) RBF
		\end{tabular}
		\fi
		\vspace{-0.4cm}
	\end{flushleft}
	\caption[Trade-off between balance and effective sample
	size]{
		Trade-off between balance and effective sample
		size. The black dashed-line indicates the estimated elbow
		point.}
	\label{fig:ndim_ss}
\end{figure}

Figure~\ref{fig:ndim_ss} illustrates the trade-off between the balance
measure (the normed difference-in-means in covariates within the
weighted subset) and effective subset size, as the balance-sample size
frontier. Such graphs can be useful to researchers in selecting a
solution along the regularization path for estimating the ATE.  Across
all cases, we achieve a good amount of balance improvement once the
data set is pruned to about 3,500, which occurs around where the
trade-off between subset size balance becomes less favorable. 

We also examine differences in dimension-by-dimension balance between
SVM and CARD and between SVM and KOM under the linear and polynomial
settings.  We do not conduct such a comparison for RBF, which is
infinite dimensional. Here, we consider four different SVM solutions:
the largest subset size solution whose standardized
difference-in-means in covariates was below 0.1 for all dimensions,
the solution whose effective sample size was nearest the subset size
for the other method, the solution whose normed difference-in-means in
covariates was closest to that of the other method, and the solution
occurring at the kneedle estimate for the elbow of the balance-weight
sum curve. We take the minimum-balance solution when no elbow exists,
as in the linear case.  The effective sample size is computed
according to the following Kish's formula:
\begin{equation}
	N_e \ = \ \frac{\left(\sum_{i\in \cT} \alpha_i\right)^2}{\sum_{i\in \cT} \alpha_i^2} + \frac{\left(\sum_{i\in \cC} \alpha_i\right)^2}{\sum_{i\in \cC} \alpha_i^2}.
\end{equation}

\begin{figure}[t!]
	\centering
	\includegraphics[width=\textwidth]{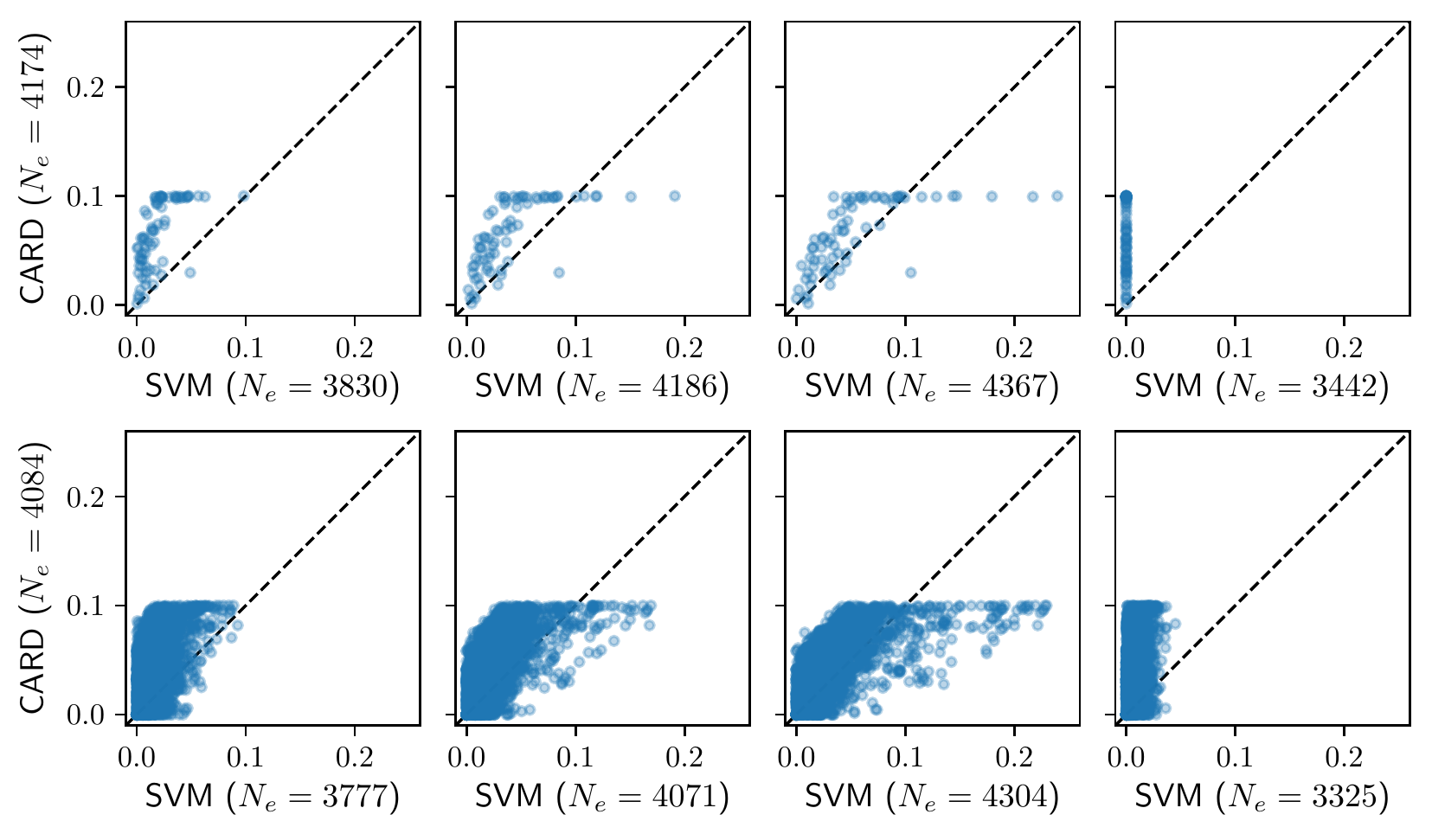}
	\begin{flushleft}
		\vspace{-0.5cm}
		\ifprint
		\else
		\begin{tabular}{C{0.75cm}  C{3.3cm}  C{3.3cm}  C{3.3cm}  C{3cm} }
			& \parbox[t]{3.2cm}{\small (a) Small difference- \\ \hphantom{(a)} in-means} & \parbox[t]{3.2cm}{\small (b) Closest effective \\ \hphantom{(b)} sample size} & \parbox[t]{3.2cm}{\small (c) Closest normed \\ \hphantom{(c)} difference-in-\\ \hphantom{(c)} means}	 & (d) Elbow
		\end{tabular}
		\fi
		\vspace{-0.4cm}
	\end{flushleft}
	\caption[Comparison of covariate standardized difference-in-means
	between SVM and cardinality matching]{
		Comparison of covariate standardized difference-in-means
		between SVM and cardinality matching (CARD) under the linear (top)
		and polynomial (bottom) settings with different SVM solutions: (a)
		standardized difference-in-means less than 0.1 in all covariates, (b) effective
		sample size closest to that of CARD, (c) normed
		difference-in-means closest to that of CARD, and (d) elbow of the
		regularization path.  The effective sample size for each method is
		given as $N_e$ in the parentheses. Note that darker areas correspond to higher concentrations of points.}
	\label{fig:sdim_card}
\end{figure} 

Figure~\ref{fig:sdim_card} presents the covariate balance comparisons
between SVM and CARD for both linear and polynomial
settings. Comparing against the small difference-in-means solution
(leftmost column) for which the standardized difference-in-means for
all covariates are below 0.1, CARD retains more observations in its
selected subset, but SVM achieves a better covariate balance than CARD
for most dimensions although there are some large imbalances.  This is
expected because SVM minimizes the overall covariate imbalance without
a constraint on each dimension as in CARD.  We observe a similar
result when comparing CARD with the SVM solutions based on the closest
effective sample size solution (left-middle) and the closest normed
difference-in-means solution (right-middle).  It is notable that the
latter generally achieves a better covariate balance while retaining
more observations than CARD.  Finally, the results for the elbow SVM
solution (rightmost column) show that tight balance is attainable with
a moderate amount of sample pruning, with near exact balance in the
linear setting.  This level of covariate balance is difficult to
achieve with CARD due to the infeasibility of optimization
particularly in high dimensional settings.

\begin{figure}[t!]
	\centering
	\includegraphics[width=\textwidth]{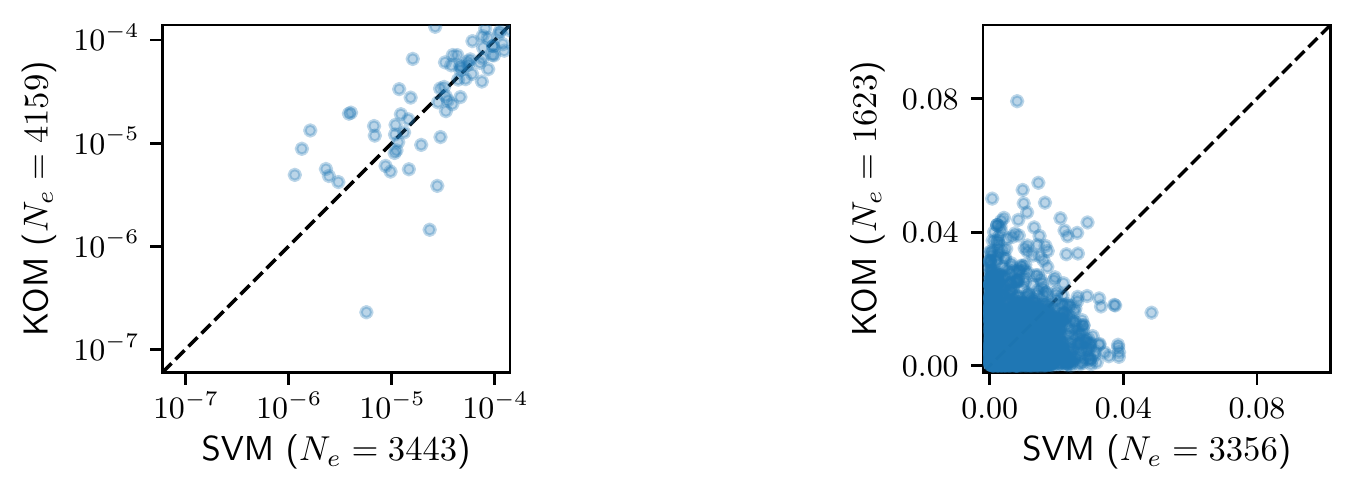}
	\begin{flushleft}
		\vspace{-0.5cm}
		\ifprint
		\else
		\begin{tabular}{C{1.5cm}  C{3.6cm}  C{5.15cm}  C{3.6cm} }
			& \small(a) Linear & & \small (b) Polynomial
		\end{tabular}
		\fi
		\vspace{-0.4cm}
	\end{flushleft}
	\caption[Comparison of covariate standardized difference-in-means
	between SVM and kernel optimal matching]
	{Comparison of covariate standardized difference-in-means between
		SVM and kernel optimal matching (KOM) under the linear
		(left) and polynomial (right) settings.}
	\label{fig:sdim_kom}
\end{figure}

Figure~\ref{fig:sdim_kom} shows the dimensional balance comparisons
for the KOM solution against the SVM solution. Here, we see that under
the linear setting, KOM retains significantly more units than the SVM
solution while attaining the same balance. This is due to the
$\sum_i \alpha W_i = 0$ constraint of SVM, which encourages the
selected subset to have a roughly equal proportion of treated and
control units, while the KOM solution allows the resulting subset to
be more imbalanced. Under the polynomial setting, however, we find
that SVM retains significantly more units than KOM while achieving a
similar degree of covariate balance, which is likely due to poor
regularization parameter choice by the KOM algorithm.

\begin{table}[t]
	\centering
	\begin{tabular}{llccc}
		\toprule
		Feature & Method & Estimate & Standard error & Effective sample size \\
		\midrule
		\multirow{5}{*}{Linear} &
		$\text{SVM}_{\text{balance}}$ & 0.0615 & 0.0148 & 3442  \\
		& $\text{SVM}_{\text{elbow}}$ & ---    & ---    & --- \\
		& $\text{SVM}_{\text{imbalance}}$  & 0.0479 & 0.0132 & 4367 \\			
		& CARD                       & 0.0335 & 0.0138 & 4174 \\
		& KOM                        & 0.0656 & 0.0147 & 4159 \\
		\midrule
		\multirow{5}{*}{Polynomial} & 
		$\text{SVM}_{\text{balance}}$ & 0.0634 & 0.0279 & 1087 \\
		& $\text{SVM}_{\text{elbow}}$ & 0.0588 & 0.0154 & 3325 \\
		& $\text{SVM}_{\text{imbalance}}$             & 0.0541 & 0.0135 & 4375 \\
		& CARD                       & 0.0313 & 0.0139 & 4084 \\
		& KOM                        & 0.0452 & 0.0251 & 1623 \\
		\midrule
		\multirow{5}{*}{RBF}    &   $\text{SVM}_{\text{balance}}$ & 0.0518 & 0.0289 & 1123 \\
		& $\text{SVM}_{\text{elbow}}$  & 0.0527 & 0.0148 & 3444 \\
		& $\text{SVM}_{\text{imbalance}}$ & 0.0474 & 0.0132 & 4378 \\
		& KCB                        & 0.0337 & 0.0173 & 3306 \\
		& KOM                        & 0.0582 & 0.0166 & 4185 \\
		\bottomrule
	\end{tabular}
	\caption[The estimated effect of right heart catheterization on death
	within 30 days after treatment]{
		The estimated effect of right heart catheterization on death
		within 30 days after treatment.  We compare the results
		based on cardinality matching (CARD) and Kernel Optimal
		Matching (KOM) with those based on three SVM solutions --
		the solution with the best covariate balance on the
		regularization path in terms of normed difference-in-means
		($\text{SVM}_{\text{balance}}$), the elbow solution
		($\text{SVM}_{\text{elbow}}$), and the solution with the
		worst covariate balance on the path
		($\text{SVM}_{\text{imbalance}}$).}
	\label{tab:rhc_results}
\end{table}

Lastly, we compare the point estimates of the ATE, the weighted Neyman
standard error, and the effective sample size for SVM, CARD (with
linear and polynomial), KOM, and KCB (with RBF) in
Table~\ref{tab:rhc_results}. We considered three different solutions
from the SVM path: $\text{SVM}_{\text{imbalance}}$, which corresponds
to the initial solution for which the balance constraint is most
relaxed and $\alpha_i = 1$, $i \in \cT$,
$\text{SVM}_{\text{balance}}$, which corresponds to the most
regularized solution with the best covariate balance on the path, and
$\text{SVM}_{\text{elbow}}$, which corresponds to the solution
occurring at the elbow of the balance-weight sum curves shown in
Figure~\ref{fig:ndim_ss}.

The results show that SVM leads to a positive estimate in all cases,
which agrees with the original finding reported in
\cite{connors1996effectiveness}. We also find that the three SVM
solutions differ most significantly in their standard errors, which
increases as the constraint on balance becomes stronger and the subset
is more pruned, as shown in the effective sample size column. In particular, 
the heavily balanced $\text{SVM}_{\text{balance}}$ solution under the RBF 
setting leads to a 
95\% confidence interval which overlaps with zero. This is in contrast with
the less balanced $\text{SVM}_{\text{elbow}}$ solution, which has both a 
larger effect estimate and smaller confidence interval. This result 
demonstrates the necessity of computing the regularization path so that 
researchers may avoid low-quality solutions due to poor parameter choice.  

Comparing against other methods, we observe that KOM yields greater
estimates of positive effects with comparable standard errors in the
linear and RBF settings. However, under the polynomial setting, the
standard error is much larger and the sample is significantly more
pruned than the modestly balanced $\text{SVM}_{\text{elbow}}$
solution.  Both CARD and KCB produce smaller positive effect
estimates, with the standard error for KCB leading to a 95\%
confidence interval which overlaps with zero.

\section{Concluding Remarks}

In this paper, we show how support vector machines (SVMs) can be used
to compute covariate balancing weights and estimate causal effects. We
establish a number of interpretations of SVM as a covariate balancing
procedure. First, the SVM dual problem computes weights which minimize
the MMD while simultaneously maximizing effective sample size. Second,
the SVM dual problem can be viewed as a continuous relaxation of the
largest balanced subset problem. Lastly, similar to existing kernel
balancing methods, SVM weights minimize the worst-case bias due to
prognostic score imbalance. Additionally, path algorithms can be used
to compute the entire set of SVM solutions as the regularization
parameter varies, which constitutes a balance-sample size
frontier. 

Our work suggests several possible directions for future research. On
the algorithmic side, a disadvantage of the proposed methodology is
that it encourages roughly equal effective number of treated and
control units in the optimal subset, which can lead to unnecessary
sample pruning.  One could use weighted SVM \citep{lin2002fuzzy} to
address this problem, but existing path algorithms are applicable only
to unweighted SVM.  On the theoretical side, our results suggest a
fundamental connection between the support vectors and the set of
overlap. \citet{steinwart2004sparseness} shows that the fraction of
support vectors for a variant of the SVM discussed here asymptotically
approaches the measure of this overlap set, suggesting that SVM may be
used to develop a statistical test for the overlap assumption.

\spacingset{1.8}
\pdfbookmark[1]{References}{References}
\bibliography{Bib,my,imai}

\newpage
\pdfbookmark[1]{Supplementary Appendix}{Supplementary Appendix}
\spacingset{1.25}
\begin{center}
  \Large Supplementary Appendix for ``Estimating Average Treatment
  Effects with Support Vector Machines''
\end{center}
 
\begin{appendix}
	\section{Proof of Theorem~\ref{thm:svm_mmd_equiv}}
	\label{app:thm:svm_mmd_equiv}
	
	We prove this theorem by establishing several equivalent
	reformulations of the SVM dual problem. By equivalence, we mean that
	these problems differ from one another only in the scaling of the
	regularization parameter, implying that their regularization paths
	consist of the same set of solutions. More formally, given two
	optimization problems P1 and P2, we say that P1 and P2 are equivalent
	if a solution $\balpha_1^*$ for P1 can be used to construct a solution
	for P2.
	
	Denote the SVM weight set
	\begin{equation*}
		\label{eq:svm_weights}
		\mathcal{A}_{\text{SVM}} = \left\{\balpha \in \R^N : 0 \preceq \balpha.
		\preceq 1, ~\sum_{i \in \mathcal{T}} \alpha_i = \sum_{j \in
			\mathcal{C}} \alpha_j\right\},
	\end{equation*}
	and consider a rescaled version of the SVM dual given in
	Eqn~\eqref{eq:svm_dual}, which we label as P1:
	\begin{equation}
		\label{eq:svm_dual_resc}
		\begin{aligned}
			\min_{\balpha} &&& \balpha^\top\bQ\balpha -  \mu \bone^\top\balpha. & \\
			\text{s.t.} &&& \balpha \in \mathcal{A}_{\text{SVM}} &
		\end{aligned}
		\tag{P1}
	\end{equation}
	Note that for a given $\lambda$ and solution $\balpha_{*}$ to the
	original problem defined in Eqn~\eqref{eq:svm_dual} under $\lambda$,
	$\balpha_{*}$ is also a solution to the rescaled problem
	\ref{eq:svm_dual_resc} under $\mu = 2\lambda$.  This establishes
	equivalence between these two problems.
	
	We begin by proving the following lemma, which allows us to replace
	the squared seminorm term $\balpha^\top \bQ \balpha$ with
	$\sqrt{\balpha^\top \bQ \balpha}$ to obtain the problem
	\begin{equation}
		\label{eq:svm_dual_sqrt}
		\begin{aligned}
			\min_{\balpha} &&& \sqrt{\balpha^\top\bQ\balpha} -  \nu \bone^\top\balpha. & \\
			\text{s.t.} &&& \balpha \in \mathcal{A}_{\text{SVM}} &
		\end{aligned}
		\tag{P2}
	\end{equation}

	\begin{lemma}
		\label{lem:sqrt_equiv}
		The problems given in Eqn~\eqref{eq:svm_dual_resc} and
                Eqn~\eqref{eq:svm_dual_sqrt} are equivalent.
	\end{lemma}
	\begin{proof}
		This result follows from the strong duality of the SVM dual problem,
		which allows us to form the following equivalent problem in which
		the penalized term $\bone^\top \balpha$ is replaced with a hard
		constraint with threshold $\epsilon$:
		\begin{equation}
			\label{eq:svm_dual_const}
			\begin{aligned}
				\min_{\balpha} &&& \balpha^\top\bQ\balpha. & \\
				\text{s.t.} &&& \balpha \in \mathcal{A}_{\text{SVM}} & \\
				&&& \bone^\top\balpha \geq \epsilon
			\end{aligned}
			\tag{P3}
		\end{equation}
		The solution to Eqn~\eqref{eq:svm_dual_const} is unchanged whether we
		minimize $\balpha^\top\bQ\balpha$ or $\sqrt{\balpha^\top\bQ\balpha}$,
		so the problem
		\begin{equation*}
			\begin{aligned}
				\min_{\balpha} &&& \sqrt{\balpha^\top\bQ\balpha} & \\
				\text{s.t.} &&& \balpha \in \mathcal{A}_{\text{SVM}} & \\
				&&& \bone^\top\balpha \geq \epsilon
			\end{aligned}
		\end{equation*}
		is identical to Eqn~\eqref{eq:svm_dual_const}. By strong duality, we can
		again enforce the hard constraint on the term $\bone^\top\balpha$
		through a penalized term with new regularization parameter, which
		establishes the equivalence between Eqn~\eqref{eq:svm_dual_resc} and
		Eqn~\eqref{eq:svm_dual_sqrt}.
	\end{proof}
	
	Next, we consider the fractional program 
	\begin{equation}
		\label{eq:svm_dual_frac}
		\begin{aligned}
			\min_{\balpha} &&& \frac{\sqrt{\balpha^\top\bQ\balpha}}{\bone^\top\balpha / 2}.  & \\
			\text{s.t.} &&& \balpha \in \mathcal{A}_{\text{SVM}} &
		\end{aligned}
		\tag{P4}
	\end{equation}
	The following lemma connects Eqn~\eqref{eq:svm_dual_frac} to
        the reformulated SVM problem given in Eqn~\eqref{eq:svm_dual_sqrt} through Dinkelbach's method \citep{dinkelbach1967nonlinear,schaible1976fractional}:
	\begin{lemma}{\citep[][Theorem~1]{dinkelbach1967nonlinear}}
		\label{lem:dinkelbach}
		Suppose $\balpha_* \in \cA_{\text{SVM}}$ and $\balpha_* \neq \bzero$. Then
		\begin{equation*}
			q_* = \frac{\sqrt{\balpha_*^{\top}\bQ\balpha_*}}{\bone^\top\balpha_*/2} = 
			\min_{\balpha \in \mathcal{A}_{\text{SVM}}} \frac{\sqrt{\balpha^\top\bQ\balpha}}{\bone^\top\balpha / 2} 
		\end{equation*}
		if, and only if
		\begin{equation*}
			\min_{\balpha \in \mathcal{A}_{\text{SVM}}} \sqrt{\balpha^\top\bQ\balpha} - \frac{q_*}{2}\bone^\top\balpha = \sqrt{\balpha_*^\top\bQ\balpha_*} - \frac{q_*}{2}\bone^\top\balpha_* = 0.
		\end{equation*}
	\end{lemma}
	
	Thus, the solution to the rescaled SVM dual problem
	defined in Eqn~\eqref{eq:svm_dual_sqrt} under $\nu = q_* / 2$ minimizes the
	fractional program given in Eqn~\eqref{eq:svm_dual_frac}. Finally, we consider the
	MMD minimization problem defined in Eqn~\eqref{eq:mmd_opt_orig},
	\begin{equation}
		\label{eq:mmd_opt_app}
		\begin{aligned}
			\min_{\balpha} &&& \sqrt{\balpha^\top\bQ\balpha}.  & \\
			\text{s.t.} &&& \balpha \in \mathcal{A}_{\text{simplex}} &
		\end{aligned}
		\tag{P5}
	\end{equation}
	The following lemma establishes equivalence between
	Eqn~\eqref{eq:svm_dual_frac} and Eqn~\eqref{eq:mmd_opt_app} under the proper
	renormalization of the fractional program solution.
	\begin{lemma}
		\label{lem:frac_equiv}
		Assume any solution $\balpha_*$ to \eqref{eq:svm_dual_frac} is such that $\balpha_* \neq \bzero.$ Then the problems \eqref{eq:svm_dual_frac} and \eqref{eq:mmd_opt_app} are equivalent.
	\end{lemma}
	\begin{proof}
		Let $\balpha_4 \neq \bzero$ and $\balpha_5$ be
                solutions to problems defined in Eqn~\eqref{eq:svm_dual_frac} and Eqn~\eqref{eq:mmd_opt_app}, respectively, and consider the vector-valued function $f : \cA_{\text{SVM}} \setminus \bzero \mapsto \cA_{\text{simplex}}, ~ f(\balpha) =
		\balpha / (\bone^\top\balpha / 2)$, which normalizes the weights in the treated and control groups to each sum to $1$. First note that since $\cA_{\text{simplex}} \subset \cA_{\text{SVM}} \setminus \bzero$, $\balpha_5$ is feasible for Eqn~\eqref{eq:svm_dual_frac}. Then by optimality of $\balpha_4$, we have
		\begin{equation*}
			\frac{\sqrt{\balpha_4^\top\bQ\balpha_4}}{\bone^\top\balpha_4/2} \leq \frac{\sqrt{\balpha_5^\top\bQ\balpha_5}}{\bone^\top\balpha_5/2} = \sqrt{\balpha_5^\top\bQ\balpha_5}.
		\end{equation*}
		Next, note that $f(\balpha_4)$ is feasible for Eqn~\eqref{eq:mmd_opt_app}. Then by optimality of $\balpha_5$, we have
		\begin{equation*}
			\sqrt{\balpha_5^\top\bQ\balpha_5} \leq \sqrt{f(\balpha_4)^\top\bQ f(\balpha_4)} = \frac{\sqrt{\balpha_4^\top\bQ\balpha_4}}{\bone^\top\balpha_4/2}.
		\end{equation*}
		In order for both of these inequalities to be true, we must have
		\begin{equation*}
			\frac{\sqrt{\balpha_4^\top\bQ\balpha_4}}{\bone^\top\balpha_4/2} = \frac{\sqrt{\balpha_5^\top\bQ\balpha_5}}{\bone^\top\balpha_5/2} = 		\sqrt{\balpha_5^\top\bQ\balpha_5}.
		\end{equation*}
	\end{proof}
	
	Note that the assumption $\balpha \neq \bzero$ in Lemma \ref{lem:dinkelbach} and Lemma \ref{lem:frac_equiv} holds when $\nu \geq q_* / 2$, provided that the data does not consist of samples all belonging to the same class. To see this, note that when $\nu = q_* / 2$, both $\balpha = \bzero$ and any scaled version of the MMD-minimizing solution, $\balpha = c\balpha_*, c > 0$, lead to an objective value of 0 in Eqn\eqref{eq:svm_dual_sqrt}. Then since Eqn~\eqref{eq:svm_dual_sqrt} is a strictly monotonically decreasing function of $\nu$, we must have $\balpha \neq \bzero$ for $\nu > q_* / 2$. Now since the regularization path defined in this work only considers $\nu \geq q_* / 2$, we may safely assume $\balpha \neq \bzero$.
	
	We are now ready to prove Theorem~\ref{thm:svm_mmd_equiv}.  Part (i):
	Lemma~\ref{lem:sqrt_equiv} establishes equivalence between the
	regularization paths for the rescaled SVM dual
	defined in Eqn~\eqref{eq:svm_dual_resc} and Eqn~\eqref{eq:svm_dual_sqrt}.  In addition,
	Lemma~\ref{lem:dinkelbach} establishes the existence of $\nu_*$ such
	that the solution to Eqn~\eqref{eq:svm_dual_sqrt} under $\nu_*$ is also a
	solution to Eqn~\eqref{eq:svm_dual_frac}. Then, it follows that there
	exists $\lambda_*$ such that the solution to the rescaled SVM dual
	problem under $\lambda_*$ minimizes Eqn~\eqref{eq:svm_dual_frac}. Finally,
	recall that Lemma~\ref{lem:frac_equiv} establishes that the minimizing
	solution to Eqn~\eqref{eq:svm_dual_frac} is also a solution to the
	weighted MMD minimization problem.  Therefore, there exists
	$\lambda_*$ such that the solution to the SVM dual under $\lambda_*$
	minimizes the weighted MMD.  Part (ii): The proof follows from
	\citet[Lemma 3]{schaible1976fractional}.

\section{Relation to Cardinality Matching and Stable Balancing Weights}
\label{app:othermethods}

Cardinality matching is an optimization procedure that maximizes the 
number of matches subject to a set of covariate balance constraints.  
The optimization problem for cardinality matching is given by,
\begin{mini}[2]
	{m_{ij}}{\sum_{i \in \mathcal{T}} \sum_{j \in \mathcal{C}} m_{ij}}{\label{eq:cardinality}}{}
	\addConstraint{\abs*{\sum_{i \in \mathcal{T}} \sum_{j \in \mathcal{C}} m_{ij} [f_b(X_{id}) - f_b(X_{jd})]}}{\leq \varepsilon_{db} \sum_{i \in \mathcal{T}} \sum_{j \in \mathcal{C}} m_{ij}, \quad}{d = 1,\dots, D, ~~ b = 1,\dots, B}
	\addConstraint{\sum_{j \in \mathcal{C}} m_{ij}}{\leq 1,}{i \in \mathcal{T}}
	\addConstraint{\sum_{i \in \mathcal{T}} m_{ij}}{\leq 1,}{j \in \mathcal{C}}
	\addConstraint{m_{ij}}{\in \{0,1\},}{i \in \mathcal{T}, ~~ j \in \mathcal{C},}
\end{mini}
where $m_{ij}$ are selection variables indicating whether treated unit
$i$ is matched to control unit $j$, $X_{id}$ denotes the $d$th element
of covariate vector $\bX_i$, $f_b$ is an arbitrary function of the
covariates specifying each of the $B$ balance conditions, and
$\varepsilon_{db}$ is a tolerance selected by a researcher. Common
choices for $f_b$ are the first- and second-order moments, and
$\varepsilon_{db}$ is typically set to a scalar multiple of the
corresponding standardized difference-in-means.

To establish the connection between SVM and cardinality matching,
we first note that cardinality matching need not be formulated as a
matched pair optimization problem.  In fact, the balance constraints
between pairs, as formulated in Eqn~\eqref{eq:cardinality}, are
equivalent to those between the treatment and control groups in the
selected subsample.  Similarly, the one-to-one matching constraints
are equivalent to restricting the number of treated and control units
in the selected subsample to be equal. Therefore, defining the
indicator variable $\alpha_i$ for selection into the optimal subset,
we can rewrite the optimization problem for cardinality matching as,
\begin{mini}[2]
	{\balpha}{\frac{1}{2}\sum_{i = 1}^N \alpha_i }{}{}
	\addConstraint{\abs*{\sum_{i \in \mathcal{T}} \alpha_i f_b(X_{id}) - \sum_{j \in \mathcal{C}} \alpha_j f_b(X_{jd})}}{\leq \frac{1}{2} \varepsilon_{db} \sum_{i = 1}^N \alpha_i, \quad}{d = 1,\dots, D, ~~ b = 1,\dots, B}
	\addConstraint{\sum_{i=0}^{N} \alpha_i W_i}{= 0,}{}
	\addConstraint{\alpha_i}{\in \{0,1\},}{i = 1, \dots, N.}
\end{mini}

Comparing this problem to the SVM dual defined in Eqn~\eqref{eq:svm_dual}, we see two
differences. First, cardinality matching restricts $\alpha_i$ to be
integer-valued, while SVM allows for $\alpha_i$ to be continuous.
Second, balance in the optimal subset is enforced differently in each
method. Cardinality matching imposes covariate-specific balance by
bounding each dimension's difference-in-means, while SVM imposes
aggregated balance by penalizing the normed difference-in-means. The
preference between these two measures of balance may in part depend on
the dimensionality of the covariates and a priori knowledge about
confounding mechanisms. If we suspect certain covariates to be strong
confounders, then bounding those specific dimensions may be
reasonable. However, if no such information is available and the
covariate space is high-dimensional, then restricting the overall
balance may be preferable.


Closely related to cardinality matching are stable balancing weights
(SBW), which is a weighting method that aims to minimize the
dispersion of the weights subject to a set of balance conditions. The
optimization problem for SBW is given by
\begin{mini}[2]
	{\balpha}{\norm{\balpha}^2}{}{}
	\addConstraint{\abs*{\sum_{i \in \mathcal{T}} \alpha_i f_b(X_{id}) - \sum_{j \in \mathcal{C}} \alpha_j f_b(X_{jd})}}{\leq \varepsilon_{db}, \quad}{d = 1,\dots, D, ~~ b = 1,\dots, B}
	\addConstraint{\balpha \in \cA_{\text{simplex}},}{}{}
\end{mini}
where $f_b(\cdot)$ and $\varepsilon_{db}$ denote the same quantities 
defined in cardinality matching.

In order to connect SBW to SVM, we consider L2-SVM, a variant of SVM 
obtained by replacing $\xi_i$ with $\xi_i^2$ in the objective function of
the SVM primal problem given in
Eqn~\eqref{eq:svm_primal}. L2-SVM has a corresponding dual form
\begin{mini}[2]
	{\balpha}{\left(\frac{1}{2}\balpha^\top\bQ\balpha - \bone^\top\balpha\right) + \frac{\lambda}{2}\norm{\balpha}^2}{\label{eq:l2svm_dual}}{}
	\addConstraint{\bW^\top\balpha}{=0}{}
	\addConstraint{\balpha}{\succeq 0.}{}
\end{mini}
Here, the term contained in parentheses relates to the MMD, while the
second term relates to the dispersion of the renormalized weights $\balpha/(\bone^\top\balpha/2)$. As before, the main difference between
L2-SVM and SBW lies in the way balance is enforced, with the former encouraging
aggregate balance through normed difference-in-means and the latter
using dimension-specific balance through a constraint.

	\section{Conditional Bias with Respect to SATE and SATT}
	\label{app:bias}
	
	In this section, we derive the conditional bias for the weighted difference-in-means estimator. Note that our derivation follows the one given in \cite{kallus2018more}. Consider the problem of estimating the SATE and SATT, defined as
	\begin{equation*}
		\tau_{\text{SATE}} = \frac{1}{N} \sum_{i=1}^N Y_i(1) - Y_i(0) ~\text{and}~ \tau_{\text{SATT}} = \frac{1}{n_T} \sum_{i \in \mathcal{T}} Y_i(1) - Y_i(0),
	\end{equation*}
	respectively. We denote the weighted estimator $\widehat{\tau}$, which has a form
	\begin{equation*}
		\widehat{\tau} = \sum_{i \in \mathcal{T}} \alpha_i Y_i - \sum_{i \in \mathcal{C}} \alpha_i Y_i,
	\end{equation*}
	where $\alpha_i \in \cA_{\text{simplex}}$. The conditional bias with respect to the SATE is given by
	\begin{align*}
		\E[\widehat{\tau} - \tau_{\text{SATE}} \mid &\bX_{1:N}, T_{1:N}] \\
		& = \sum_{i \in \mathcal{T}} \E[\alpha_i Y_i \mid \bX_{1:N}, T_{1:N} ] - \sum_{i \in \mathcal{C}} \E[\alpha_i Y_i \mid \bX_{1:N}, T_{1:N}] - \E[\tau_{\text{SATE}} \mid \bX_{1:N}, T_{1:N}] \\
		&= \sum_{i = 1}^N \alpha_i [T_i - (1 - T_i)]\E[ Y_i(T_i) \mid \bX_i, T_i ] - \frac{1}{N}\sum_{i=1}^{N} \E[Y_i(1) - Y_i(0) \mid \bX_i, T_i] \\
		&= \sum_{i = 1}^N \alpha_i [T_i - (1 - T_i)]\E[ Y_i(T_i) \mid \bX_i] - \frac{1}{N}\sum_{i=1}^{N} \E[Y_i(1) - Y_i(0) \mid \bX_i] \\
		&= \sum_{i = 1}^N \alpha_i T_i f_1(\bX_i) - \sum_{i = 1}^N \alpha_i (1 - T_i) f_0(\bX_i) - \frac{1}{N}\sum_{i=1}^{N} \tau(\bX_i) \\
		&= \sum_{i = 1}^N \alpha_i T_i [f_0(\bX_i) + \tau(\bX_i)] - \sum_{i = 1}^N \alpha_i (1 - T_i) f_0(\bX_i) - \frac{1}{N}\sum_{i=1}^{N} \tau(\bX_i) \\
		&= \sum_{i = 1}^N (\alpha_i T_i - N^{-1}) \tau(\bX_i) + \sum_{i=1}^{N} \alpha_i[T_i - (1-T_i)] f_0(\bX_i) \\
		&= \sum_{i = 1}^N (\alpha_i T_i - N^{-1}) \tau(\bX_i) + \sum_{i=1}^{N} \alpha_iW_i f_0(\bX_i),
	\end{align*}
	where the second equality follows from SUTVA, and the third equality follows from
	Assumptions~\ref{assn:uncon}~and~\ref{assn:over}. By a similar argument, 
	the conditional bias with respect to the SATT is
	\begin{align*}
		\E[\widehat{\tau} - \tau_{\text{SATT}} \mid &\bX_{1:N}, T_{1:N}] \\
		&= \sum_{i \in \mathcal{T}} \E[\alpha_i Y_i \mid \bX_{1:N}, T_{1:N} ] - \sum_{i \in \mathcal{C}} \E[\alpha_i Y_i \mid \bX_{1:N}, T_{1:N}] - \E[\tau_{\text{SATT}} \mid \bX_{1:N}, T_{1:N}] \\
		&= \sum_{i \in \mathcal{T}}^N (\alpha_i - n_T^{-1}) \tau(\bX_i) + \sum_{i=1}^{N} \alpha_iW_i f_0(\bX_i).
	\end{align*}
	
	\section{Worst-case Bias in an RKHS}
	\label{app:worst-case}
	
	We consider the problem of minimizing the bias due to prognostic score imbalance, defined in Eqn~\eqref{eq:bias_cate}. Restricting $f_0$ to the unit-ball RKHS, defined as $\mathcal{F}_K$ in Section~\ref{subsec:mmd}, and considering the $f_0$ which maximizes the absolute value of this quantity, we compute the worst-case squared bias due to prognostic score imbalance as
	\begin{equation*}
		B^2(\balpha; ~\bX_{1:N}, T_{1:N}) = 
		\sup_{f_0 \in \mathcal{F}_K} \left(\sum_{i=1}^{N} \alpha_i W_i f_0(\bX_i) \right)^2.
	\end{equation*}
	We can simplify this expression by
	\begin{align*}
		\label{eq:worst_bias}
		B^2(\balpha; \bX_{1:N},T_{1:N}) &= 
		\sup_{f_0 \in \mathcal{F}_K} \left(\sum_{i=1}^{N} \alpha_i W_i f_0(\bX_i) \right)^2 \\
		&= \sup_{f_0 \in \mathcal{F}_K} \left(\sum_{i=1}^{N} \alpha_i W_i \ip*{f_0}{\phi(\bX_i)} \right)^2 \\
		&= \sup_{f_0 \in \mathcal{F}_K} \left(\ip*{f_0}{\sum_{i=1}^{N} \alpha_i W_i\phi(\bX_i)} \right)^2 \\
		&= \norm{\textstyle\sum_{i=1}^{N} \alpha_i W_i\phi(\bX_i)}_{\mathcal{H}_K}^2, \\
		&= \gamma_K^2 \left(\widehat{F}_{\balpha}, \widehat{G}_{\balpha} \right),
	\end{align*}
	where the second line follows from the reproducing property of
        the RKHS, the third line follows from bilinearity of inner
        products, and the fourth line follows from the Cauchy-Schwarz
        inequality.

\end{appendix}

\end{document}